\documentclass[%
reprint,
amsmath,amssymb,
aps,
pra, 
superscriptaddress
]{revtex4-1}

\usepackage{graphicx}
\usepackage{dcolumn}
\usepackage{bm}
\usepackage{multirow}

\usepackage{amsmath,amssymb,amsfonts}
\usepackage{graphicx}
\usepackage{color}
\usepackage{appendix}
\usepackage{amsthm}
\usepackage{tikz}
\usepackage{youngtab}
\usepackage{mleftright}
\usepackage{url}

\newtheorem{theorem}{Theorem}

\newcommand{\nn}{\nonumber}

\def\a{\alpha}
\def\b{\beta}
\def\g{\gamma}
\def\d{\delta}

\def\th{\theta}

\def\la{\lambda}

\def\n{\nu}

\def\r{\rho}

\def\s{\sigma}

\def\p{\psi}
\def\w{\omega}

\def\P{\Psi}

\def\<{\langle}
\def\>{\rangle}
\def\Tr{Tr}

\def\Tr{{\rm Tr}}
\def\tr{{\rm tr}}

\def\jmath{{j}}

\def\perm{{\textrm{perm}}}
\def\mN{{\mathcal{N}}}

\begin{document}

	\title{Entanglement of Identical Particles
		 and Coherence\\
		 in the First Quantization Language}

	\author{Seungbeom Chin }
	\email{sbthesy@gmail.com}
\affiliation{Molecular Quantum Dynamics and Information Theory Laboratory, Department of Chemistry, Sungkyunkwan University, Suwon 16419, Korea}

\author{Joonsuk Huh}%
\email{joonsukhuh@gmail.com}
\affiliation{Molecular Quantum Dynamics and Information Theory Laboratory, Department of Chemistry, Sungkyunkwan University, Suwon 16419, Korea}
\affiliation{SKKU Advanced Institute of Nanotechnology (SAINT), Sungkyunkwan University, Suwon 16419, Korea}

	
	\begin{abstract} 
We suggest a formalism to illustrate the entanglement of identical particles in the first quantization language (1QL). Our 1QL formalism enables one to exploit all the well-established quantum information tools to understand the indistinguishable ones, including the reduced density matrix and familiar entanglement measures. The rigorous quantitative relation between the amount of entanglement and the spatial coherence of particles is possible in this formalism. Our entanglement detection process is a generalization of the entanglement extraction protocol for identical particles with mode splitting proposed by Killoran et al. (2014).
	\end{abstract}
	
	
	\maketitle

\section{Introduction} \label{intro}

Entanglement is a fundamental quantum feature that cannot be imitated by classical systems, and also works as a useful informational resource.
Entanglement causes EPR \cite{einstein1935can} and Bell \cite{bell2004speakable} paradoxes that reveal the non-local property of quantum systems, and enables several quantum tasks such as quantum  teleportation \cite{bennett1993teleporting} and many quantum algorithms \cite{horodecki2009quantum}.
There is a consensus that the entanglement is based on the superposition of distinctive multipartite states \cite{ schrodinger1935discussion}.

On the other hand, it arises confusion to apply this interpretation of the entanglement to the entanglement for identical particles.
A set of identical particles satisfies the symmetrization principle \cite{dirac1981principles}, which is represented as the exchange symmetry of the corresponding wave function among the particles in the first quantization language (1QL). This symmetrization results in the superposition of multipartite states. Just focusing on the mathematical structure of the states, one could assume that the particles are strongly entangled \cite{ichikawa2008exchange,wei2010exchange}.
However, since the particle indistinguishability prevents any observer from addressing the individual particles, it is meaningless to discuss the entanglement of the total system by  taking the particles as subsystems. Many works claim that such entanglement is just a mathematical artifact that is unphysical and useless (not a quantum resource) \cite{ ghirardi1977gc,  ghirardi2004general, schliemann2001j,schliemann2001ja, eckert2002k, esteve2008j, ghirardi2002entanglement, paskauskas2001r, zanardi2002p, shi2003shi, barnum2004subsystem, barnum2005generalization, tichy2013entanglement}. A common viewpoint of these works is that the particle labels in 1QL, which are unphysical for the case of identical particles, cause the confusion. Hence, to see whether there exists realistic and useful entanglement that comes from the particle indistinguishability, we should be able to discard the unphysical entanglement (a mathematical illusion) from the physical one.

One possible attempt is to play in the second quantization language (2QL, a mode-based approach) \cite{benatti2011entanglement,benatti2012bipartite,benatti2012entanglement,   marzolino2015quantum,marzolino2016performances,benatti2017remarks}, which only involves mode creation operators without signifying the individual particle labels. The separability in 2QL is related to commuting algebras of observables, instead of the Hilbert space tensor structure. Quantities such as the negativity and robustness of entanglement are suggested as criterions for discriminating the separability of given states. 
Recently, an unorthodox approach is suggested \cite{franco2016quantum,sciara2017universality,bellomo2017,compagno2018dealing, giuseppe2018, castellini2018entanglement}, which is named the non-standard approach (NSA) by the authors. It is a particle-based language (similar to 1QL), however without imposing pseudo-labels to the particles (similar to the 2QL). 
It provides a tool to identify the quantitative relation between the spatial overlap of identical particles and entanglement.   
With NSA, one can define the partial trace of particle states \cite{franco2016quantum} and Schmidt-decompose the states universally \cite{sciara2017universality}, with which we can define familiar entanglement measures such as von Neumann entropy and concurrence (a recent work showed that this NSA can be recovered in 2QL \cite{lourencco2019entanglement}).

In this paper, we show that 1QL approach can accomplish the same tasks as the former approaches (and more), by extending the symmetrization principle to subsystems that involves some part of the whole particles. With the \emph{symmetrized partial trace} (a partial trace for identical particles), we can define the reduced density matrix for respective subsystems and familiar entanglement measures such as von Neumann entropy.
By using 1QL, we can exploit many familiar concepts in quantum information, e.g., coherence, to quantitatively understand the entanglement generation process of identical particles.
Our formalism also enables us to interpret the entanglement  extraction protocol of identical particles suggested in Ref. \cite{killoran2014extracting} in a more generalized viewpoint.
The authors of Ref. \cite{killoran2014extracting} showed that the particle-based symmetrization entanglement can be extracted (or transferred) to the mode entanglement, which is a useful resource for quantum information tasks (a fermionic example of this type of protocol can be found in Ref. \cite{cavalcanti2007useful}). 
The case discussed in Ref. \cite{killoran2014extracting} is restricted to completely overlapping identical particles with internal pseudospins in space.  Here we treat more generalized situations when particles are partially overlapping. Our result shows that \emph{the spatial coherence of identical particles is a necessary (but not sufficient) condition that the entanglement of the identical particles are detected.} Theorem \ref{eandc} of Section \ref{pisc} clarifies the conditions for the spatial coherence of identical particles and the location of detectors to extract the entanglement of identical particles.

This work is organized as follows:  
Section \ref{1QL} introduces a 1QL formalism for the entanglement of identical particles with internal degrees of freedom. By defining the partial trace of multipartite systems for the case of identical particles, we can quantify the amount of the entanglement. In Section \ref{pisc}, we show that spatial coherence is an imperative element for a given set of identical particles with internal pseudospins to be entangled. Here the computational basis of spatial coherence is determined by the distinguishable detectors, which compose the bipartite entangled systems. Two simplest but informative examples ($N=2$ and $N=3$) are also analyzed in Section \ref{example}. In Section \ref{killoran}, we reinterpret the entanglement extraction protocol by mode splitting proposed by Killoran et al. \cite{killoran2014extracting} using our formalism. We can state that our formalism is a more generalized case of that in Ref. \cite{killoran2014extracting}. In Section \ref{conclusions}, we summarize our discussion and suggest some possible researches in the future. 

\section{1QL formalism for the entanglement of identical particles}\label{1QL}

We can find the key requirements for the definition of the entanglement of identical particles in Ref. \cite{dalton2017quantum}. First, the subsystems are distinguishable from each other (individually accessible by physical observers). Second, measurements are made on the subsystems. Third, the subsystems can exist as separate quantum systems (hence the symmetrization of particle states itself does not correspond to an actual entanglement, since it prohibits any form of separable states except when all particles are in the same state). Only when these conditions are satisfied, we can state the entanglement is $physical$. And since modes are distinguishable and particles are not, \emph{a mode (or a collection of modes) can be a subsystem} for our case, not particles.

The above requirements seem to imply that 2QL is more suitable for the entanglement of identical particles than 1QL since 2QL indicates only modes but not individual particles. On the other hand, it is usually considered that the pseudo-labels of particles accompanied by 1QL scramble the physical and fictitious entanglements.
However, as we will show in this section, it is possible to illustrate the physical entanglement of identical particles with 1QL (by applying the symmetrization principle not only to particles but also to detectors). Besides, once it is achieved, the quantification of entanglement with the reduced density matrix for one subsystem is also possible, as in NSA \cite{franco2016quantum, sciara2017universality}.
This implies that, by staying in 1QL, one can exploit the well-established entanglement formalism of distinguishable particles to analyze the case of indistinguishable particles.
 


\begin{figure}[t]
	\centering
	\includegraphics[width=8cm]{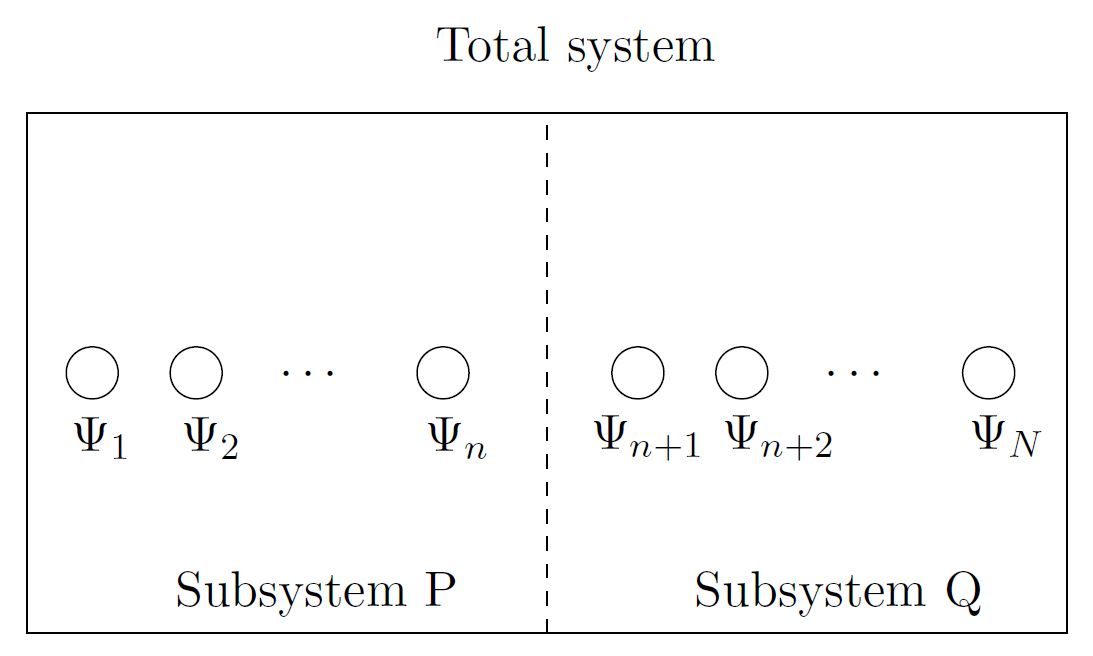}
	\caption{$N$ identical particles which are devided into two subsystems $P$ and $Q$. An observer in $P$ can know that $n$ identical particles have states $(\P_1,\cdots,\P_n)$, and one in $Q$ can know that  $(N-n)$ identical particles have states $(\P_{n+1},\cdots,\P_N)$ }
	\label{twosub}
\end{figure}

Let us consider $N$ instinguishable bosons contained in a physical system (when we state that some particles are in a physical system, it implies that the particles have potential to be detected by all the observers who can classically communicate with each other).
In 1QL, the wavefunction of a particle with pseudo-label A  with a state $\P$ contained in the system is decribed by $|\P\>_A$ (note that the pseudo-label is a mere mathematical tool for identical particles, which cannot be addressed individually), where
$\P$ includes information on the spatial distribution $\p$ and internal state of the particles $s$, i.e., $\P = (\p,s)$. Then the transition amplitude of $|\P\>_A$ is given by
\begin{align}\label{trans}
_A\<\Phi|\P\>_B = \<\Phi|\P\> \d_{AB}.
\end{align} 
 
The total state of $N$ identical particles with pseudo-labels $A_a$ ($a=1,\dots, N$) and states $\P_i$ ($i=1,\dots,N$) is described as follows: 
\begin{align}\label{totalstate}
&|\P_1,\P_2, \cdots , \P_N\> \nn \\ &= \mathcal{N}_N(\n)\Big[ \sum_{\s} |\P_{1}\>_{A_{\s(1)}}|\P_{2}\>_{A_{\s(2)}}\cdots |\P_{N}\>_{A_{\s(N)}}\Big],
\end{align}
where the summation is over all possible permutations $\s$, which represents the exchange symmetry among particles. $\mN _N(\n)$ is the normalization factor that depends on the number of particles that are in the same state. More specifically, when 
\begin{align}
(\P_1,\P_2,\cdots, \P_N) = (\underbrace{x_1,\cdots, x_1}_{\n_1}, \underbrace{x_2,\cdots, x_2}_{\n_2},\cdots,\underbrace{x_l,\cdots, x_l}_{\n_l} ),
\end{align}
$\mN_N(\n)$ is written as
\begin{align}
 \mN_N(\n)= \sqrt{\frac{\prod_{j=1}^{l}\n_j! }{ N!} }.
\end{align} 
Using Eq. \eqref{trans} and \eqref{totalstate},
it is direct to see that the transition amplitude from $|\P_1,\cdots , \P_N\>$ to $|\Phi_1,\cdots , \Phi_N\>$ is given by
\begin{align}\label{ta}
& \< \Phi_1,\cdots , \Phi_N|\P_1,\cdots , \P_N\>\nn \\ & = \mN_N(\Phi)\mN_N(\P) \sum_{\r,\s}\prod_{i=1}^{N}\<\Phi_{\r(i)}|\P_{\s(i)}\>.
\end{align}
By defining a matrix $A$ so that its entries are imposed as $A_{ij} =\<\Phi_i|\P_j\> $, Eq. \eqref{ta} is rewritten as 
\begin{align}\label{perm}
  \< \Phi_1,\cdots , \Phi_N|\P_1,\cdots , \P_N\> =\frac{1}{N!\mN_{N}(\Phi)\mN_N(\P) } \perm(A),
\end{align} where 
`perm' represents the matrix permanent (for the case of fermions, Eq. \eqref{totalstate} experiences sign changes in the summation along the degree of permutation, hence the transition amplitude is proportional to the matrix determinant of $A$).

It is slightly subtle to consider the states of subsystems that contain $n$ ($\le N$) identical particles (see Fig. \ref{twosub}). While the particle states can be identified, e.g., as $(\P_1,\dots,\P_n)$, the particle pseudolabels are unphysical quantities that cannot be detected in principle. Hence, the partial wavefunction $|\P_1,\cdots,\P_n\>$ for the subsystem also should be symmetrized with respect to the particle labels. We can achieve the symmetric states by arranging the wavefuctions in the form of elementary symmetric polynomials, i.e., 
\begin{align}\label{partsym}
|\P_1,\cdots,\P_n\> &= \mathcal{N}_{[1]}(\P)\sum_{\substack{a_1\neq a_2 \\ \neq \cdots\neq a_n} }|\P_1\>_{A_{a_1}} \cdots |\P_{n}\>_{A_{a_n}}, 
\end{align} where $1\le a_p \le N$ for all $p=1,\cdots, n$.
\begin{widetext}
For example, when $n=1$, the wavefunction is a simple summation over all pseudo-labeled particles, 
\begin{align}\label{k1}
 |\P\> = \mN_{[1]}(\P) \Big( |\P\>_{A_1}+ |\P\>_{A_2}+ \dots +|\P\>_{A_N}\Big). \quad \Big(\mN_{[1]}(\P) =  \frac{1}{\sqrt{N}} \Big)  
\end{align}
When $n=2$ and $N=3$, we have
\begin{align} 
|\P_1,\P_2\> = \mN_{[2]}(\P) \Big(|\P_1\>_{A_1}|\P_2\>_{A_2}+ |\P_1\>_{A_1}|\P_2\>_{A_3}+ |\P_1\>_{A_2}|\P_2\>_{A_1}  +|\P_1\>_{A_2}|\P_2\>_{A_3}+|\P_1\>_{A_3}|\P_2\>_{A_1}+|\P_1\>_{A_3}|\P_2\>_{A_2}\Big),
\end{align}
where $\mN_{[2]}(\P)=1/\sqrt{3!}$ ($1/\sqrt{3}$) when $\P_1 \neq \P_2$ ($\P_1 = \P_2$).
\end{widetext}
Eq. \eqref{partsym} renders us to define the \emph{symmetrized partial trace} of a given density marix $\r$ over a subsystem S with $n$ identical particles. 
If the identity matrix for $S$ is expressed with 
\begin{align}\label{IS}
\mathbb{I}_S = \sum_{a}|\Phi_{1}^a, \cdots,\Phi_{n}^a\> \<\Phi_{1}^a,\cdots, \Phi_{n}^a|
\end{align} ($\{|\Phi_{1}^a, \cdots,\Phi_{n}^a\>\}_a$ composes the complete symmetric computational basis set for the subsystem $S$),
so that it satisfies $\mathbb{I}_S|\P_1,\cdots,\P_N\> =|\P_1,\cdots,\P_N\>$,
the symmetrized partial trace over $S$ is defined by
\begin{align}\label{partial}
	\Tr_S(\r) \equiv \sum_a \<\Phi_{1}^a, \dots,\Phi_{n}^a| \r |\Phi_{1}^a,\dots, \Phi_{n}^a\>.
\end{align}
Since the whole wave functions contained in Eq.~\eqref{partial} are symmetric under the particle pseudolabels, the resultant reduced density matrix is symmetric without doubt.   
With such a reduced density matrix we can evaluate the amount of entanglement, e.g., by defining von Neumann entropy, which is detectable and physical (it is worth mentioning that  the fictitous entanglement of identical articles is from the pseudo-label symmetry-breaking of states).
Even if Eq.~\eqref{partial} is defined with pure states of identical particles, the extension of the argument so far to a mixed state case $\r= \sum_{a}\p_a|\P_a\>\<\P_a|$ is straightforward.
It is also not very hard to show that the symmetrized states Eq. \eqref{totalstate} and  Eq. \eqref{partsym} are equivalent to the label-absent holistic particle states defined in NSA \cite{compagno2018dealing}. See Appendix \ref{NSA} for a more detailed explanation. 
  
Now we check whether the entanglement we have discussed satisfies the three conditions mentioned at the beginning of this section. First, the whole state $|\P_1,\cdots, \P_N\>$ is divided into two systems that contain $n$ and $(N-n)$ identical particles respectively, which are distinguishable by the spatial distribution and the detectable quantum states. Second, the choice of $\mathbb{I}_S$ (Eq. \eqref{IS}) depends on our way of measurement, therefore susceptible to measurements. Third, the subsystems can exist as separate systems when the reduced density matrix Eq. \eqref{partial} has only one nonzero eigenvalue.

One can understand the symmetrized partial trace from the viewpoint of Hopf algebra in the $N$-particle Hilbert space.  It was argued in Ref. \cite{balachandran2013entanglement,balachandran2013algebraic} that the partial trace method for delving into the entanglement of distinguishable particles is not applicable to the entanglement of indistinguishable particles, for partial traces on a specific Hilbert space is not equivalent to restrictions to subalgebras of obsevables on the totally symmetrized Hilbert space. However, by symmetrizing the observables, we can see that the symmetrized partial traces correspond to restrictions to subalgebras \cite{Chin}.  
  
We have shown so far that all known flaws of 1QL when analyzing the entanglement of identical particles can be overcome by defining symmetrized subsystem states as Eq. \eqref{partsym}. Therefore, we can explicitly evaluate the entanglement of a given set of identical particles with apparent wavefunction combination \footnote{This technic contrasts with NSA \cite{compagno2018dealing}, which performs all calculations with implicitly defined projection relations among states}, examining the time evolution of the particles as well. We can see that in the first quantization approach every algebra we need is derived directly from the symmetrized particle states themselves without any extra definition. Another important advantage of the first quantization approach is that the exchange symmetry of bosons and fermions naturally can be extended to a more generalized case, i.e., a mixture of symmetry and antisymmetry among different particles. While the transition amplitudes of bosons and fermions are expressed as matrix permanents (Eq. \eqref{perm}) and determinants, it is possible to consider particles whose transition amplitudes are proportional to immanants. Such kind of particles are named $immanons$ \cite{tichy2017extending}. We expect to apply the discussion in the current section to understanding the correlation of immanons in the future.

\section{From spatial coherence to the bipartite entanglement of identical particles}\label{pisc}

It is widely admitted that a set of non-interacting identical particles can generate entanglement only after their spatial wave functions overlap \cite{killoran2014extracting, franco2016quantum, sciara2017universality, paunkovic2004role} (``No quantum prediction, referring to an atom located in our laboratory, is affected by the mere presence of similar atoms in remote parts of the universe \cite{peres2006quantum}''). An interesting thought experiment that clarifies this point is given in Section II of Ref. \cite{compagno2018dealing}. 
However, as we will show in this section, a mere spatial overlap among identical particles is not a prerequisite for the entanglement, for it cannot identify the relation of particles with detectors \footnote{Both the concept of detection-level entanglement \cite{tichy2013entanglement} and the algebraic approach to the  entanglement of indistinguishability \cite{benatti2011entanglement,benatti2012bipartite,benatti2017remarks} presume that the entanglement generated from  indistinguishability depends on the experimental context, which implies not just the spatial overlap among identical particles but also between particles and dete ctors.}. 

Here we show how the spatial overlap of particles with two separate detectors determines the amount of bipartite entanglement, based on the method developed in the former section. 
The locations of two detectors define the computational basis for the spatial coherence of identical particles. We present a rigorous condition for the state to be entangled along the variation of spatial coherence. 

Our focus is on $N$ identical bosons with pseudospin up and down, which can be applied to the Bose-Einstein condensation (BEC) of cold atoms or  scattering of photons with polarizations (we can compare our results with those of  Ref. \cite{killoran2014extracting} and  Ref. \cite{giuseppe2018} in this septup).
Supposing among $N$ bosons $n$ particles are in pseudospin up ($ \uparrow$) and $(N-n)$ are in pseudospin down ($\downarrow$), the total wave function is given by
\begin{align}\label{nk}
|\P\> &= |\P_1,\P_2,\dots, \P_N\> \nn \\
&=|\p_1 \uparrow,\dots, \p_n\uparrow, \p_{n+1}\downarrow,\dots,\p_{N}\downarrow\>,
\end{align}
where $\p_j$ ($1\le j\le N$) are the spatial wave functions for the particles (note that from Eq. \eqref{totalstate} an wave function written in this form is symmetrized).


\begin{figure}[t]
	\centering
	\includegraphics[width=8.5cm]{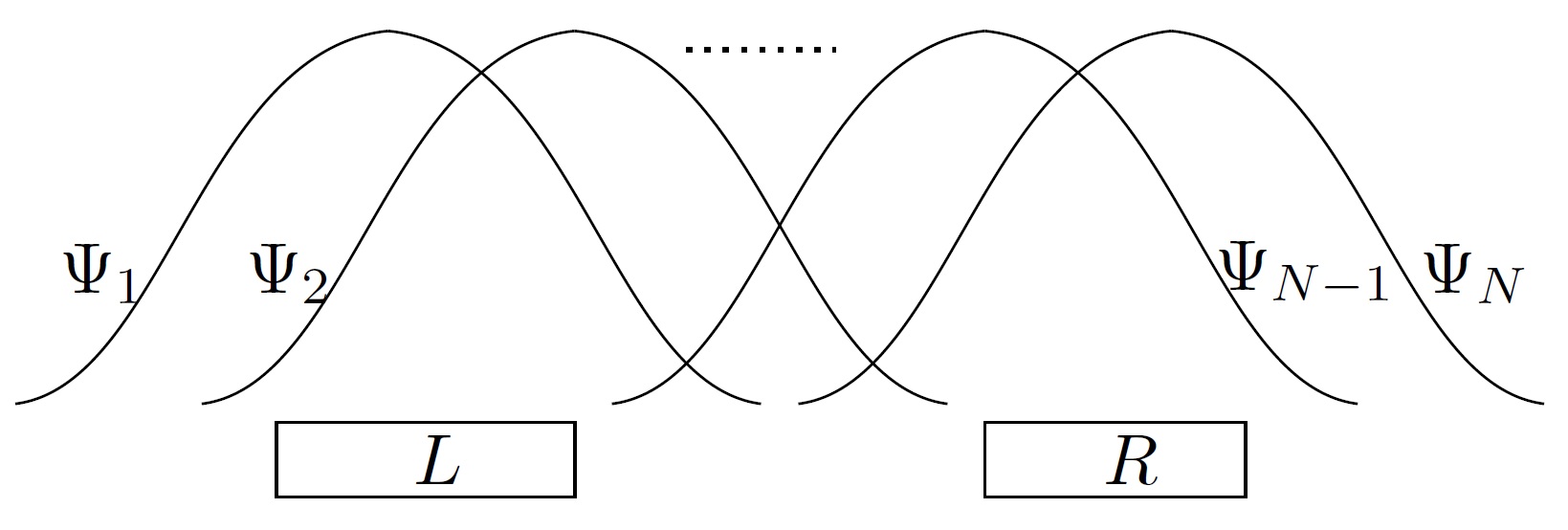}
	\caption{$N$ identical particles arrives at two detectors $L$ and $R$}
	\label{nidpart}
\end{figure}

\subsection*{Spatial coherence of identical particles}

The particles can be detected (spatially overlapped) by detector $L$, or by detector $R$, or not detected by both (see Fig. \ref{nidpart}).
Hence, we express each spatial wave function $|\p_j\>$ in the most general form  as 
\begin{align}
|\p_j\> = \sin\phi_j\big( \cos\th_j|L\> + e^{i\w_j}\sin\th_j|R\>\big)  + e^{i\g_j}\cos\phi_j|\chi_j\>,
\end{align}
where $0\le \th_j$, $\phi_j \le \pi/2$, $\<L|R\>=0$ (two detectors are far enough to be distinguished) and  $ |\chi_j\>$ are orthogonal to both $|L\>$ and $|R\>$. It is easy to see that the angles $\phi_j$ and $\th_j$ determine the probabilities that the particle arrives at $L$, or at $R$, or does not arrive at all.
Then the spatial overlap between two states $|\p_j\>$ and $|\p_k\>$ are quantified as
\begin{align}
\<\p_j|\p_k\> =& \sin\phi_j\sin\phi_k (\cos\th_j\cos\th_k + e^{i(\omega_k-\omega_j)}\sin\th_j\sin\th_k) \nn \\
&+e^{i(\gamma_k - \gamma_j)}\cos\phi_j\cos\phi_k.
\end{align}
Note that $\p_j$ and $\p_k$ can overlap even when the corresponding particles  arrive at neither of the detectors $L$ and $R$ ($\phi_j=\phi_k = 0$). For this case, it is meaningless to discuss the entanglement of the particles between modes $L$ and $R$. Therefore, we can see that the spatial overlap itself is not directly related to the entanglement of identical particles. 

Instead, we focus on the relation of particles with detectors, which can be describe with the coherence of the particles in the computational basis set $\{|L\>, |R\> \}$. 
Since the values of $\phi_j$ do not affect the physical state by the normalization after the projection, we can set $\phi_j=\pi/2$ without loss of generality and 
\begin{align}\label{generalstate}
|\p_j\> = \cos\th_j|L\> + e^{i\w_j}\sin\th_j|R\>.
\end{align}
And the density matrix $\r_j$ for the state $|\p_j\>$ is given by
\begin{align}
\r_j = 
\begin{pmatrix}
\cos^2\th_j & e^{i\w_j}\cos\th_j\sin\th_j \\
e^{-i\w_j}\cos\th_j\sin\th_j & \sin^2\th_j
\end{pmatrix}.
\end{align} 
Then  the amount of coherence $C_j$ for $|\p_j\>$ is evaluated by the off-diagonal elements \cite{baumgratz2014quantifying}. For our case, $\r_j$ has  only one independent element, $C_j$ can be directly defined as 
\begin{align}\label{coher}
C_j = 2\cos\th_j\sin\th_j.
\end{align} (if $|\p_j\>$ has more than two bases, we need to define some scalar quantities to evaluate the amount of coherence according to the axioms for coherence measures and monotones. See, e.g., Ref. \cite{streltsov2017colloquium}). $C_j$ is normalized so that $0\le C_j\le 1$. Note that $C_j=0$ when $\th_j=0$ or $\pi/2$, and  $C_j=1$ when $\th_j=\pi/4$. 

One common property of entanglement of identical particles and coherence is that they are both basis-dependent (detector dependent), which is an indirect evidence that two quantites are closely related.

\subsection*{From coherence to entanglement of $N$ identical particles with pseudospins} 
Now we examine how the coherence defined as Eq. \eqref{coher} affects the entanglement.
The observation of particles by the detectors $|R\>$ and $|L\>$ is represented as the projection of $|\P\>$ on   $\Pi_{LR}=\sum_{\a,\b}|\Phi^{\a\b}\>\<\Phi^{\a\b}|$ ($0 \le \a\le n$, and $0 \le \b \le N-n$), where
\begin{align}\label{pab}
|\Phi_{\a\b}\> &\equiv|\Phi^{\a\b}_1, \dots , \Phi^{\a\b}_N \> \nn \\ 
&= |\underbrace{L\uparrow, \dots}_{\a}, \underbrace{L\downarrow, \dots}_{\b},  \underbrace{R\uparrow, \dots}_{n-\a}, \underbrace{R\downarrow,\dots}_{N-n-\b}\>
\end{align} 
($|\Phi_{\a\b}\>$ is defined in a form that preserves the particle number of spin up and down).
Then the projected state by the detectors is given with Eq. \eqref{perm} by
\begin{align}
|\P_{LR}\> &=\sum_{q=1}^N 
\sum_{\a+\b=q}\frac{\perm(A^{\a\b})}{N!\mN_N(\Phi^{\a\b})\mN_N(\Phi)} |\Phi^{\a\b}\> \nn \\
&\equiv \sum_{q=1}^N |\Phi_q\>, 
\end{align}
where $A^{\a\b}$ is a $N\times N$ matrix whose entries are given by $(A^{\a\b})_{jk}= \<\Phi^{\a\b}_j|\P_k\>$. We have introduced $|\Phi_q\>$ so as to group the projected states into those which have equal particle number distributions between two systems, to regard the particle number superselection rule \cite{wiseman2003entanglement}. 
Using  $\<L|\p_j\> = \cos\th_j$ and $\<R|\p_j\>= e^{i\w_j}\sin\th_j$, $A^{\a\b}$ is expressed in the matrix form as
\[
\renewcommand\arraystretch{1.3}
\mleft[
\begin{array}{c|c|c|c}
C_{\a} & 0 & S_{n-\a} & 0 \\
\hline
0 & C_{\b} & 0 & S_{N-n-\b}
\end{array}
\mright]
\]
where

	\begin{align}\label{hab}
	&C_\a= 
	\underbrace{\begin{pmatrix}
		\cos\th_1  &\cdots& \cos\th_1 \\
		\vdots& \vdots &\vdots \\
		\cos\th_n &\cdots & \cos\th_n
		\end{pmatrix}
	}_{\a}, 
	\nn \\
	&C_{\b} = \underbrace{\begin{pmatrix}
		\cos\th_{n+1}  &\cdots& \cos\th_{n+1} \\
		\vdots	& \vdots &\vdots \\
		\cos\th_N &\cdots & \cos\th_N
		\end{pmatrix}
	}_{\b}, \nn \\ 
	& S_{n-\a }= 
	\underbrace{\begin{pmatrix}
		e^{i\w_1}\sin\th_1  &\cdots& e^{\w_1}\sin\th_1 \\
		\vdots& \vdots &\vdots \\
		e^{i\w_n}\sin\th_n &\cdots & e^{i\w_n}\sin\th_n
		\end{pmatrix}
	}_{n-\a} ,
	\nn \\
	 &S_{N-n-\b}= 
	\underbrace{\begin{pmatrix}
		e^{i\w_{n+1}}\sin\th_{n+1}  &\cdots& e^{i\w_{n+1}}\sin\th_{n+1} \\
		\vdots& \vdots &\vdots \\
		e^{i\w_{N}}\sin\th_N &\cdots & e^{i\w_{N}}\sin\th_N
		\end{pmatrix}
	}_{N-n-\b}.
	\end{align} 

To obtain the entanglement of $|\Phi_q\>$, we need to partial trace $|\P_{q}\>\<\P_{q}|$ over the identity matrix $\mathbb{I}_L^q$ of the subsystem $L$ with 
\begin{align}
\mathbb{I}_L^q = \sum_{\a=0}^{q}|\underbrace{L\uparrow, \dots}_{\a}, \underbrace{L\downarrow, \dots}_{q-\a }\>\<\underbrace{L\uparrow, \dots}_{\a}, \underbrace{L\downarrow, \dots}_{q-\a }|.
\end{align}  
Then from the obtained reduced density matrix $\r^q_{R} = \tr_L(\mathbb{I}_L^q|\P_{q}\>\<\P_{q}|)$, we can compute the amount of the entanglement for $|\Phi_q\>$, which we denote as $E(\Phi_q)$. 
The average entanglement of $|\P_{LR}\>$ can be evaluated as
\begin{align}\label{averageent}
E(\P_{LR}) = \sum_{q}p_{q}E(\P_q).
\end{align} This is \emph{the entanglement of particles} defined in Ref. \cite{wiseman2003entanglement}, of which the local operations in $L$ and $R$ are performed by two detectors after measuring (postselecting) the particle numbers \cite{amico2008entanglement} (the name ``entanglement of particles'' means that it is about the correlation among modes with the equal particle  number distribution, not taking unaccesible individual particles as subsystem).

The spatial coherence Eq. \eqref{coher} determines whether a given state of identical bosons is entangled or separable, which can be stated as the following theorem: 
\begin{theorem} (Spatial coherence criterion for entanglement) \label{eandc} Consider the case when
 	an $N$ identical particle state $|\P\> = |\P_1,\P_2,\dots, \P_N\>=|\p_1 \uparrow,\dots, \p_n\uparrow, \p_{n+1}\downarrow,\dots,\p_{N}\downarrow\>$ is detected by two detectors $|L\>$ and $|R\>$ that locate far enough to each other, i.e., $\<L|R\>=0$. Then if the spatial coherence of idential bosons satisfy $C_1=C_2=\cdots =C_{n}=0$ or  $C_{n+1}=C_{n+2}=\cdots =C_{N}=0$, the projected state is separable.
\end{theorem}
\begin{proof}
	What we need to prove is that for all pairs of $(\a,\b)$ that satisfy $\a+\b= q$, only one pair has nonzero perm$(A^{\a\b})$, which corresponds to the separable state. 
	Since a matrix permanant is invariant under the exchange of rows, we can write
	\[
	\renewcommand\arraystretch{1.3}
	\perm (A^{\a\b})= 
	\perm \mleft[
	\begin{array}{c|c|c|c}
	C_{\a} &  S_{n-\a} & 0 & 0  \\
	\hline
	0 & 0 & C_{\b} & S_{N-n-\b}
	\end{array}
	\mright] 
	\]
	\[
	\renewcommand\arraystretch{1.3}
	\phantom{sdfdfff}\equiv 
	\perm \mleft[
	\begin{array}{c|c}
	X_{n \times n } & 0   \\
	\hline
	0  & Y_{(N-n)\times (N-n)}
	\end{array}
	\mright].
	\]
	
	$X_{n\times n}$ is a $(n\times n)$- matrix and $Y_{(N-n)\times (N-n)}$ is a $\big((N-n)\times (N-n) \big)$- matrix.
	
	First consider $C_j= 2\sin\th_j\cos\th_j = 0$ for $1\le j\le n$. This condition restricts the form of $X_{n\times n}$ so that the absolute values of the elements become $0$ or $1$.
	Suppose the permanant is nonzero  when $\a= x$. Then since a matrix permanant is invariant under the exchange of rows, without loss of generality we can set the matrix $X_{n\times n}$ as
	
	\[
	\renewcommand\arraystretch{1.3} X_{n\times n} = 
	\mleft[
	\begin{array}{c|c}
	\mathbf{1}_{x\times x} & \mathbf{0} \\
	\hline
	\mathbf{0} & \Omega_{(n-x)\times (n-x)}
	\end{array}
	\mright]
	\]
	where $(\mathbf{1}_{x\times x})_{jk} = 1$ and $ (\Omega_{(n-x)\times(n-x)})_{jk} = e^{i\w_j}$ for all $j$ and $k$. 
	Then for the fixed values of $\cos\th_j$ and $\sin\th_j$ ($1\le j \le n$), it is direct to see that the matrix permanant is zero for all $\a=x + a$ ($a\neq 0$).   This means that
	for a given $q$, there exists only one $(\a,\b)$ that contributes a nonzero amplitude. Therefore, the quantum state is separable. Setting $C_j= 2\sin\th_j\cos\th_j = 0$ for $n+1\le j\le N$ affects the form of $Y_{(N-n)\times (N-n)}$, which also result in a separable quantum state.
	
\end{proof}
 Note that the  inverse is not true, i.e., not all separable states have the spatial coherence with the above restrictions. It is because the phases $\w_j$ ($j=1,\cdots ,N$) can also affect the amount of entanglement. We can find such an example for $N=3$ case, which will be discussed  in Section \ref{example}.
\\

So far, we have discussed the quantitative relation between coherence and entanglement of given identical particles. In a sentence, nonzero coherence is a prerequisite for nonzero entanglement. As mentioned at the beginning of this section, a mere spatial overlap among identical particles does not guarantee the detection of entanglement, for the spatial relation between particles and detectors (which determines the amount of coherence) is a crucial factor as well. Even if particles spatially overlap, no entanglement is detected by placing any detector out of the overlapped region (see Eq. \eqref{coher}). This fact shows that the entanglement of particles has the property of \emph{detector-level entanglement} (introduced in Ref.~\cite{tichy2013entanglement}), which emerges from the incorporation of the particle wavefunction and the measurement process.  Appendix \ref{Potentiality} provides some interpretational discussion on this viewpoint.

\section{The simplest nontrivial examples}\label{example}

In this section we study two simplest but nontrivial cases ($N=2$ and $3$), with which one can understand the physical implications of Theorem \ref{eandc} more concretely. 
\subsection*{$N=2$}

\begin{figure}[t]
	\centering
	\includegraphics[width=6cm]{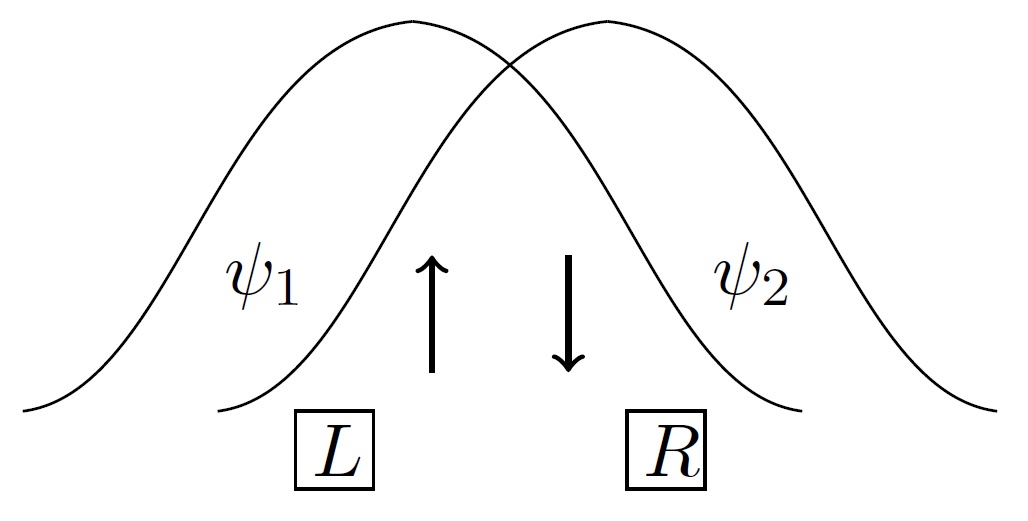}
	\caption{Two identical particles with spatial overlap. Two detector $L$ and $R$ must locate in the overlapping region of the particles to detect the entanglement.}
	\label{twooverl}
\end{figure}

First we analyze the simplest case, i.e., two bosons with internal pseudospin.
Suppose we have two identical bosons, one of which is initially in a spatial state $\p_1$ with spin up and the other is in a spatial state $\p_2$ with spin down, i.e., $\P_1 = (\p_1 \uparrow)$  and $\P_2 = (\p_2 \downarrow)$ (it is trivial to see that two particles in the same internal spin has no entanglement).
Then the total wave function is given by
\begin{align}\label{updated}
|\P_1,\P_2\> = \frac{1}{\sqrt{2}}\Big[ |\p_1,\uparrow\>_A|\p_2,\downarrow\>_B + |\p_2,\downarrow\>_A|\p_1,\uparrow\>_B \Big]
\end{align} (Fig. \ref{twooverl}).
To see the entanglement of this state, 
we prepare two detectors located distinctively, i.e., one is at a spatial region $L$ and the other is at $R$ with $\<L|R\>=0$.
Then using Eq. \eqref{hab}, the projected state is given by 
\begin{align}\label{n=2}
&|\P_{LR}\> \nn \\
&= \cos\th_1\cos\th_2|L\uparrow, L\downarrow\>  \nn \\ &\quad  + e^{i\w_2}\cos\th_{1}\sin\th_{2}|L\uparrow,R\downarrow\> + e^{i\w_1}\sin\th_{1}\cos\th_{2}|L\downarrow,R\uparrow\> \nn \\
 & \quad + e^{i(\w_1+\w_2)}\sin\th_1\sin\th_2|R\uparrow, R\downarrow\>.
\end{align}
Taking the partial trace with
$\mathbb{I}^1_L = |L\uparrow\>\<L\uparrow|+|L\downarrow\>\<L\downarrow|$ over $L$, the reduced density matrix $\r^1_R$ is given by
\begin{align}
\r^1_R&= \frac{\cos^2\th_1\sin^2\th_2 |R\downarrow\>\<R\downarrow|  + \sin^2\th_1\cos^2\th_2 |R\uparrow\>\<R\uparrow|}{\cos^2\th_2\sin^2\th_1+\cos^2\th_1\sin^2\th_2} \nn \\
&\equiv|R\>\<R|\otimes  ( \la_1^2|\uparrow\>\<\uparrow| + \la_2^2|\downarrow\>\<\downarrow|)
\end{align} (the exact same form of reduce density matrix is also obtained in Ref. \cite{giuseppe2018} using NSA).  
We can obtain entaglement measures with $\r_L^1$. e.g., von Neumann entropy ($= -\sum_{i}\la_i^2\log_2(\la_i^2)$) or entanglement concurrence ($= 2\la_1\la_2$).
    
Therefore, for both $\la_1$ and $\la_2$ not to be zero (which corresponds to the entangled state), the two detectors must locate in the overlapped region of the two particles (as in Fig. \ref{twooverl}).
The transformed state $|\P_1,\P_2\>$ has both indistinguishability and nonzero coherence, by which we can obtain non-trivial entanglement of particle now. This kind of entanglement can be used as a resource of quantum teleportation and bell inequality violation \cite{giuseppe2018,paunkovic2004role}.


Moreover, for this simplest case, the coherence of two spatial distributions $\p_{1}$ and $\p_2$ determine the amount of entanglement completely, i.e., independent of $\w_1$ and $\w_2$.  We choose concurrence as the entanglement measure for now. Indeed, with Eq. \eqref{averageent}, the average concurrence is written with $C_j$ as
\begin{align}\label{entcoh}
E_c(|\P_{LR}\>) &=  \cos\th_1\sin\th_1\cos\th_2\sin\th_2\nn \\
&= \frac{1}{4}C_1C_2.
\end{align}
Therefore the amount of entanglement for $N=2$ is completely determined by $C_1$ and $C_2$. $E_c(|\P_{LR}\>)$ has the maximal value  when $C_1=C_2=1$.
Eq. \eqref{entcoh} shows that the entanglement of particles vanishes if one of the coherences vanishes, which agrees with Theorem \ref{eandc}. This means that even if two identical particles are spatially overlapped, no entanglement is extracted when one detector places out of the overlapped region.

\subsection*{$N=3$}

Now  we move on to the 3 particle state, of which the only nontrivial pseudospin distribution is given by $|\P_1,\P_2,\P_3\>$ with $(\P_1,\P_2,\P_3)= (\p_{1}\uparrow,\p_{2}\uparrow,\p_3\downarrow)$. 
As we will see, this is the simplest case that other than spatial coherence the relative phases affect the amount of entanglement. 
Using Eq. \eqref{hab}, we can see that the nonzero $A^{\a\b}$ are ($A^{21},A^{20},A^{11},A^{01},A^{10},A^{00}$). The permanant of those matrices are given by
\begin{widetext}
\begin{align}
&\perm(A^{21})=2\cos\th_{1}\cos\th_2\cos\th_3, \nn \\ &\perm(A^{20})= 2e^{i\w_3}\cos\th_{1}\cos\th_2\sin\th_3, \quad 
\perm(A^{11})= (e^{i\w_2}\cos\th_{1}\sin\th_2 + e^{i\w_1}\sin\th_1\cos\th_2)\cos\th_3, \nn \\
&\perm(A^{01}) = 2e^{i(\w_1+\w_2)}\sin\th_{1}\sin\th_2\cos\th_3, \quad
\perm(A^{10}) = e^{i\w_3}(e^{i\w_1}\sin\th_{1}\cos\th_2 + e^{i\w_2}\cos\th_1\sin\th_2 )\sin\th_3, \nn \\
&\perm(A^{00}) = 2e^{i(\w_1+\w_2+\w_3)}\sin\th_{1}\sin\th_2\sin\th_3,
\end{align}  
Therefore, $|\P_{LR}\>$ is given by
\begin{align}\label{n3plr}
|\P_{LR}\>
 =& \cos\th_{1}\cos\th_2\cos\th_3|L\uparrow,L\uparrow,L\downarrow\>\nn \\
 &+ \big(e^{i\w_3} \cos\th_{1}\cos\th_2\sin\th_3 |L\uparrow,L\uparrow,R\downarrow\>+\frac{1}{\sqrt{2}}(e^{i\w_2}\cos\th_{1}\sin\th_2 + e^{i\w_1}\sin\th_1\cos\th_2)\cos\th_3 |L\uparrow,L\downarrow,R\uparrow\> \big)\nn \\
 & + \big(e^{i(\w_1+\w_2)}\sin\th_{1}\sin\th_2\cos\th_3 |R\uparrow,R\uparrow,L\downarrow\> +\frac{1}{\sqrt{2}}e^{i\w_3}(e^{i\w_1}\sin\th_{1}\cos\th_2 + e^{i\w_2}\cos\th_1\sin\th_2)\sin\th_3|R\uparrow,R\downarrow,L\uparrow\>\big) \nn \\
 &+ e^{i(\w_1+\w_2+\w_3)}\sin\th_{1}\sin\th_2\sin\th_3|R\uparrow,R\uparrow,R\downarrow\>
\end{align}
(note that this is not normalized).

We again choose concurrence as the entanglement measure. Considering that $|s_1,s_2\>$ ($s_1,s_2=1,2$) construct a triplet, the concurrence can be computed using Eq. (12) of Ref. \cite{chin2017coherence}. From Eq. \eqref{averageent}, the average concurrence $E_{c}(|\P_{LR}\>) = \sum_{q}p_q E_c(\P_q)$ is given by

\begin{align}
E_c(|\P\>_{LR}) =& (\cos^2\th_{1}\sin^2\th_2+\sin^2\th_1\cos^2\th_2 + 2\cos(\w_1-\w_2)\cos\th_1\sin\th_1\cos\th_2\sin\th_2)^{\frac{1}{2}} \nn \\
&\times \frac{\sin\th_3\cos\th_3(\cos\th_1\cos\th_2 + \sin\th_1\sin\th_2)}{N(\th,\w)},	
\end{align}
where $N(\th,\w)$ is a constant for normalizing Eq. \eqref{n3plr}.
Using Eq. \eqref{coher}, the above equation is rewritten as
\begin{align}
E_c(|\P\>_{LR}) = \frac{1}{N(C_i,\w)}C_3\Big(1+C_1C_2\cos(\w_1-\w_2) \pm  \sqrt{(1-C_1^2)(1-C_2^2)}\Big)^{\frac{1}{2}}\Big(1+\frac{1}{4}C_1C_2 \mp  \sqrt{(1-C_1^2)(1-C_2^2)}\Big)^{\frac{1}{2}} .
\end{align} ($\pm$ in the first parenthesis and $\mp$ in the second parenthesis are determined by the relative region of $\th_1$ and $\th_2$. In other words, $(+, -)$-sign is when both $\th_1$ and $\th_2$ are larger or smaller than $\pi/4$ and $(-,+)$ is when $\th_{1}$ is larger than $\pi/4$ and $\th_2$ is smaller than $\pi/4$ (or the other way around)).
\end{widetext}
It is insightful to compare the above result with Theorem \ref{eandc}. 
As stated in the theorm, if $C_1=C_2=0$ or $C_3=0$, the detected entanglement $E_c(\P_{LR})$ is zero. On the other hand, we can easily find a counterexample which shows that the inverse is not true, for the role of relative phases $\w_j$ are not trival now. As such an example, suppose that $C_1=C_2=1 \neq 0$ and $C_{3} \neq 0$. For this case, the total state can still be unentangled when  $\w_1-\w_2= \pi$ is satisfied.

\section{Reinterpreting the entanglement extraction protocol in Ref. \cite{killoran2014extracting}}\label{killoran}

Among various approaches to the entanglement of identical particles, Cavalcanti et al. \cite{cavalcanti2007useful}  and Killoran et al. \cite{killoran2014extracting} proposed an interesting viewpoint to match the entanglement by the symmterization of particles (based on 1QL) and the mode entanglement (based on 2QL). In this section we review the viewpoint and method of Ref. \cite{killoran2014extracting} (which focused on the bosonic case) and show that our entanglement detection process discussed so far is a generalization of the work.

As we have mentioned in Section \ref{1QL}, Eq. \eqref{totalstate} has a mathematically equivalent form with entangled states. Let us consider $N$ photons in the same mode with internal pseudospins. With $n$ particles in spin up and $(N-n)$ particles in spin down, the state is written with Eq. \eqref{nk} as
\begin{align}\label{symmstate}
|\P\> &= |\underbrace{\p\uparrow,\cdots \p\uparrow}_{n},\underbrace{\p\downarrow,\cdots,\p\downarrow}_{N-n}\> \nn \\
&= \frac{1}{\sqrt{\binom{N}{n}}}\Big[ \sum_{\s}|\uparrow\>_{A_{\s(1)}}\cdots |\uparrow\>_{A_{\s(n)}}|\downarrow\>_{A_{\s(n+1)}}\cdots |\downarrow\>_{A_{\s(N)}}\Big] \nn \\
&\equiv \frac{1}{\sqrt{\binom{N}{n} }} \mathcal{S} \Big[ |\uparrow\>_{A_1}\cdots |\uparrow\>_{A_n}\> |\downarrow\>_{A_{n+1}}\cdots |\downarrow\>_{A_{N}} \Big],
\end{align}
where the spatial wave function $\p$ is abbreviated from the second line. 
Then we can arbitrarily group the $N$ particles into two subsystems ($X,Y)$, in which $X$ ($Y$) has $N_X$ ($N_Y$) particles with pseudolabels $A_1,\cdots, A_{N_X}$ ($A_{N_X+1},\cdots ,A_N $). Along with this bipartition, Eq. \eqref{symmstate} can be rewritten as
\begin{align}
|\P\> = \frac{1}{\sqrt{\binom{N}{n}}} \sum_{n_X+n_Y=n} \Big[ \mathcal{S}|v_{n_X}\> \Big] \Big[ \mathcal{S}|v_{n_Y}\> \Big],
\end{align}  
where $|v_{n_X}\>$$=$$ |\uparrow\>_{A_1}\cdots|\uparrow\>_{A_{n_X}}|\downarrow\>_{A_{n_X+1}}\cdots |\downarrow\>_{A_{N_X}}$.
This is a Schmidt decomposed form of a bipartite quantum state, with coefficients  $\la_{n_X,n_Y} = \sqrt{\binom{N_X}{n_X}\binom{N_Y}{n_Y}/ \binom{N}{n} }$. Now this state is split into two different modes $|L\>$ and $|R\>$, so that $\p$ evolves into
\begin{align}
|\p\> \to |\p'\>= r|L\> + s|R\>. \quad (|r|^2 + |s|^2 = 1)
\end{align} 

\begin{figure}[t]
	\centering
	\includegraphics[width=7cm]{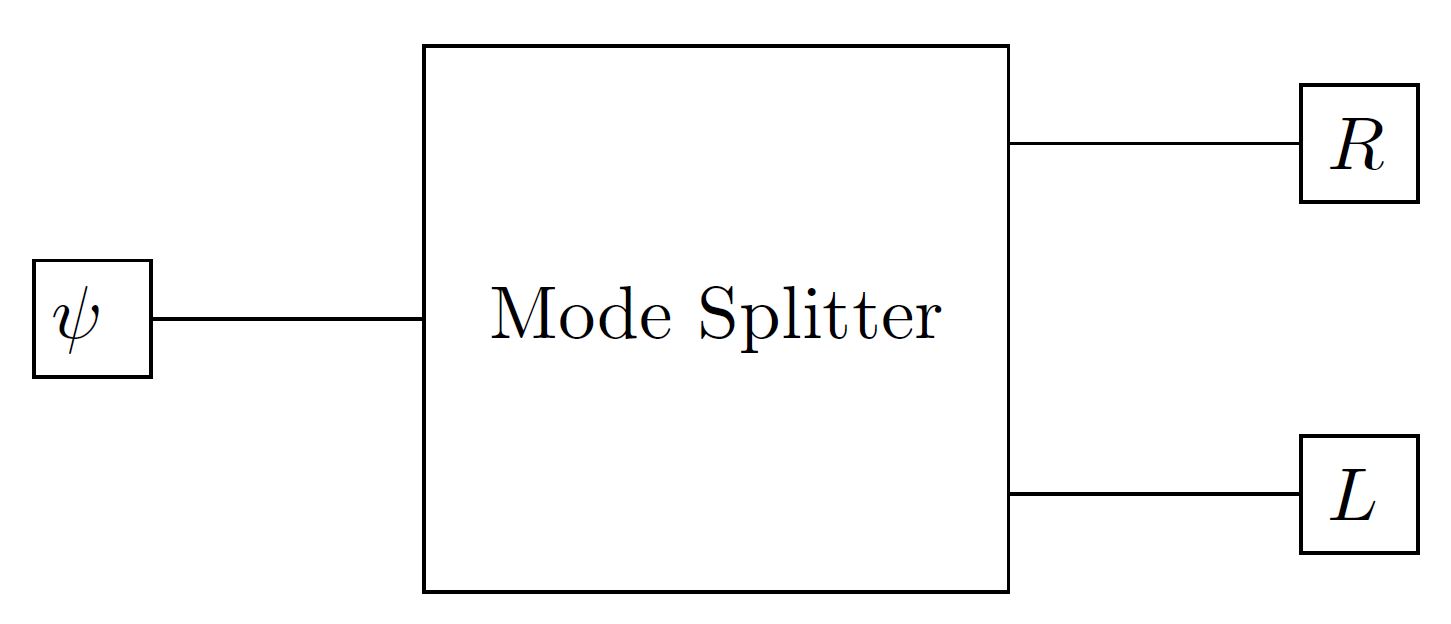}
	\caption{Mode splitting process of $N$ identical particles into two distinctive output modes $L$ and $R$. The input mode locates in position $ \p$, in which $N$ identical particles overlap with internal pseudospins. After the splitting, $N_L$ ($N_R$) particles are found in mode $L$ ($R$). }
	\label{samemode}
\end{figure}

This mode splitting process in optics is a beam splitter transformation, and in Bose-Einstein condensation (BEC) a tunneling operation (this process is represented pictorially as Fig. \ref{samemode}. The input mode is locally in position $ \p$, in which $N$ identical particles overlap with internal pseudospins). Then for the transformed state, after projection onto local particle numbers ($N_L, N_R $) = ($N_X,N_Y$), the Schmidt form of the input state is equal to that of the projected output state (the Schmidt equivalence of particle and mode states \cite{killoran2014extracting}). For example, $3$ identical particle initial state  with two of them in spin up and one in spin down is given by
\begin{align}\label{3idstate}
|2,1\> = \frac{1}{\sqrt{3}} \big[ |\uparrow\>_{1} |\uparrow\>_{2} |\downarrow\>_{3} + |\uparrow\>_{2} |\uparrow\>_{3} |\downarrow\>_{1} + |\uparrow\>_{3} |\uparrow\>_{1} |\downarrow\>_{2}  \big].
\end{align}
Suppose we allot particles 1 and 2 in $X$ and particle 3 in $Y$.  Then Eq. \eqref{3idstate} is rewritten as
\begin{align}\label{3sd}
|2,1\> =  \frac{1}{\sqrt{3}} \Big[& \Big( |\uparrow\>_1|\uparrow\>_2\Big)_X\Big( |\downarrow\>\Big)_Y \nn \\
&+  \sqrt{2}\Big(\frac{1}{\sqrt{2}}\mathcal{S}[|\uparrow\>_1|\downarrow\>_2] \Big)_X\Big( |\uparrow\>_Y \Big)\Big]. 
\end{align}
The equation has a Schmidt decomposed form with coefficients $(\la_{2,0}, \la_{1,1}) = (\sqrt{1/3},\sqrt{2/3})$. The output state after a mode splitting is easily obtained using Eq. \eqref{n3plr}. By setting  $\th_i=\th$ and $\w_i= \w$ in Eq. \eqref{n3plr}, the output state $|\P_o\>$ is given by
\begin{align}
|\P_{o}\> = &\cos^3\th|L\uparrow,L\uparrow,L\downarrow\>\nn \\
& + e^{i\a} \cos^2\th\sin\th \big[|L\uparrow,L\uparrow,R\downarrow\> +\frac{1}{\sqrt{2}}|L\uparrow,L\downarrow,R\uparrow\>\big] \nn \\
&+ e^{i\a}\sin^2\th \cos\th \big[|R\uparrow,R\uparrow,L\downarrow\> +\frac{1}{\sqrt{2}}|R\uparrow,R\downarrow,L\uparrow\>\big] \nn \\
&+ \sin^3\th|L\uparrow,L\uparrow,L\downarrow\>,
\end{align}which is equivalent to Eq. (2) of Ref. \cite{killoran2014extracting}.
For $(N_L,L_R) = (2,1)$ or $(1,2)$, we have $(\la_{2,0},\la_{1,1}) =(\sqrt{1/3},\sqrt{2/3})$ (the same Schmidt coefficients with those of Eq. \eqref{3sd}).

Killoran et al. \cite{killoran2014extracting} understood this isomorphism as the mapping of the input particle entanglement into the output mode entanglement.  This interpretation implies that the entanglement of identical particles by exchange symmetry can be extracted by passive mode splittings. In other words, the mathematical entanglement structure of many bosons is accessible with distinguishable mode subsystems.

\begin{figure}[t]
	\centering
	\includegraphics[width=6cm]{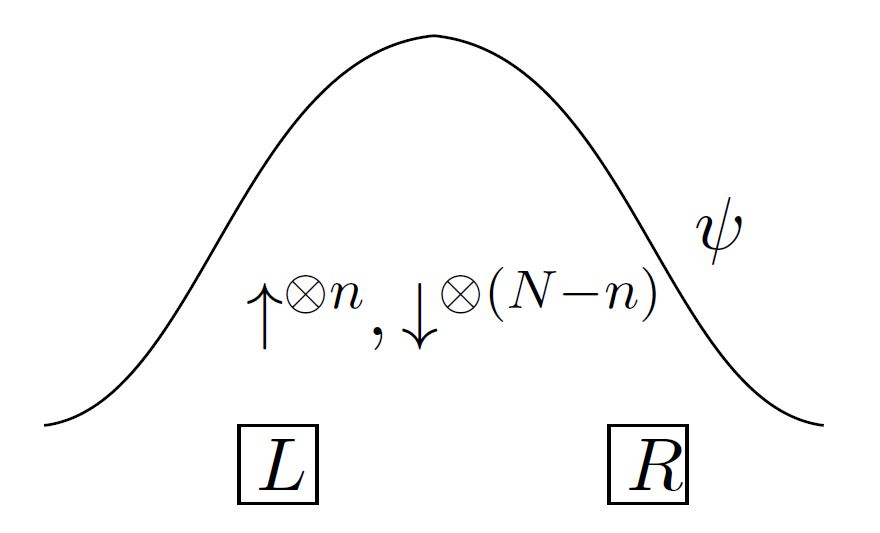}
	\caption{$n$- bosons with spin up and $(N-n)$- bosons with spin down. All $N$ bosons have an equal spatial wave function $\p$.}
	\label{overlap}
\end{figure}

This isomorphism can be restated with $N$ identical bosons that completely overlap in space and two distinguishable detectors $|L\>$ and $|R\>$ (Fig. \ref{overlap}): \emph{if the relation between the spatial wave fuction of completeley overlapping identical particles ($|\p\>$) and two detectors ($|L\>$ and $|R\>$) are given by $|\p\> = r|L\> +s|R\>$ ($|r|^2 + |s|^2=1$), there exists Schmidt equivalence of a symmetrized particle state in 1QL and a mode state in 2QL.} It is obvious to see that this isomorphism is disrupted when the particles just partially overlaps. For example, Eq. \eqref{n3plr} does not keep the input Schmidt decomposed form in general (when $\p_i$ are not equal to each other).

\section{Conclusions}\label{conclusions}
 
In this work, we formalized the first quantization approach to describe the bipartite entanglement of identical particles and proposed a criterion for unentangled states when the identical particles are spatially coherent. When $N=2$ the amount of entanglement is completely determined by coherence, and when $N\ge 3$ the relative phases also affect the entanglement. We also reinterpreted the entanglement extraction protocol of identical particles \cite{killoran2014extracting} from the viewpoint of the detector location and showed that it is a particular case when the identical particles completely overlap in space. 
 
Even if our current analysis focused on bipartite correlations, it can be extended to more general multipartite cases.  
We expect the quantitative method suggested here would apply to many quantum processes in which the entanglement of identical particles plays a central role. For example, the spin squeezing test for entanglement has remained in a qualitative domain so far \cite{dalton2017quantum2}, which would be understood better with our method containing the relation of entanglement and spatial coherence. Also, 1QL we formalized here can be applied to understanding the correlation of identical immanons \cite{tichy2017extending} (the particles of which the scattering amplitudes are proportional to immanants, not permanents (bosons) or determinants (fermions)).

\section*{Acknowledgements}

The advice of the anonymous referee helped us to sharpen our interpretation of the technical results.
This work is supported by Basic Science Research Program through the National Research Foundation of Korea (NRF) funded by the Ministry of Education, Science and Technology (NRF-2015R1A6A3A04059773).

\bibliography{firstquanti}

\begin{thebibliography}{52}%
\makeatletter
\providecommand \@ifxundefined [1]{%
 \@ifx{#1\undefined}
}%
\providecommand \@ifnum [1]{%
 \ifnum #1\expandafter \@firstoftwo
 \else \expandafter \@secondoftwo
 \fi
}%
\providecommand \@ifx [1]{%
 \ifx #1\expandafter \@firstoftwo
 \else \expandafter \@secondoftwo
 \fi
}%
\providecommand \natexlab [1]{#1}%
\providecommand \enquote  [1]{``#1''}%
\providecommand \bibnamefont  [1]{#1}%
\providecommand \bibfnamefont [1]{#1}%
\providecommand \citenamefont [1]{#1}%
\providecommand \href@noop [0]{\@secondoftwo}%
\providecommand \href [0]{\begingroup \@sanitize@url \@href}%
\providecommand \@href[1]{\@@startlink{#1}\@@href}%
\providecommand \@@href[1]{\endgroup#1\@@endlink}%
\providecommand \@sanitize@url [0]{\catcode `\\12\catcode `\$12\catcode
  `\&12\catcode `\#12\catcode `\^12\catcode `\_12\catcode `\%12\relax}%
\providecommand \@@startlink[1]{}%
\providecommand \@@endlink[0]{}%
\providecommand \url  [0]{\begingroup\@sanitize@url \@url }%
\providecommand \@url [1]{\endgroup\@href {#1}{\urlprefix }}%
\providecommand \urlprefix  [0]{URL }%
\providecommand \Eprint [0]{\href }%
\providecommand \doibase [0]{http://dx.doi.org/}%
\providecommand \selectlanguage [0]{\@gobble}%
\providecommand \bibinfo  [0]{\@secondoftwo}%
\providecommand \bibfield  [0]{\@secondoftwo}%
\providecommand \translation [1]{[#1]}%
\providecommand \BibitemOpen [0]{}%
\providecommand \bibitemStop [0]{}%
\providecommand \bibitemNoStop [0]{.\EOS\space}%
\providecommand \EOS [0]{\spacefactor3000\relax}%
\providecommand \BibitemShut  [1]{\csname bibitem#1\endcsname}%
\let\auto@bib@innerbib\@empty
\bibitem [{\citenamefont {Einstein}\ \emph {et~al.}(1935)\citenamefont
  {Einstein}, \citenamefont {Podolsky},\ and\ \citenamefont
  {Rosen}}]{einstein1935can}%
  \BibitemOpen
  \bibfield  {author} {\bibinfo {author} {\bibfnamefont {A.}~\bibnamefont
  {Einstein}}, \bibinfo {author} {\bibfnamefont {B.}~\bibnamefont {Podolsky}},
  \ and\ \bibinfo {author} {\bibfnamefont {N.}~\bibnamefont {Rosen}},\
  }\href@noop {} {\bibfield  {journal} {\bibinfo  {journal} {Physical review}\
  }\textbf {\bibinfo {volume} {47}},\ \bibinfo {pages} {777} (\bibinfo {year}
  {1935})}\BibitemShut {NoStop}%
\bibitem [{\citenamefont {Bell}(2004)}]{bell2004speakable}%
  \BibitemOpen
  \bibfield  {author} {\bibinfo {author} {\bibfnamefont {J.~S.}\ \bibnamefont
  {Bell}},\ }\href@noop {} {\emph {\bibinfo {title} {Speakable and unspeakable
  in quantum mechanics: Collected papers on quantum philosophy}}}\ (\bibinfo
  {publisher} {Cambridge university press},\ \bibinfo {year}
  {2004})\BibitemShut {NoStop}%
\bibitem [{\citenamefont {Bennett}\ \emph {et~al.}(1993)\citenamefont
  {Bennett}, \citenamefont {Brassard}, \citenamefont {Cr{\'e}peau},
  \citenamefont {Jozsa}, \citenamefont {Peres},\ and\ \citenamefont
  {Wootters}}]{bennett1993teleporting}%
  \BibitemOpen
  \bibfield  {author} {\bibinfo {author} {\bibfnamefont {C.~H.}\ \bibnamefont
  {Bennett}}, \bibinfo {author} {\bibfnamefont {G.}~\bibnamefont {Brassard}},
  \bibinfo {author} {\bibfnamefont {C.}~\bibnamefont {Cr{\'e}peau}}, \bibinfo
  {author} {\bibfnamefont {R.}~\bibnamefont {Jozsa}}, \bibinfo {author}
  {\bibfnamefont {A.}~\bibnamefont {Peres}}, \ and\ \bibinfo {author}
  {\bibfnamefont {W.~K.}\ \bibnamefont {Wootters}},\ }\href@noop {} {\bibfield
  {journal} {\bibinfo  {journal} {Physical review letters}\ }\textbf {\bibinfo
  {volume} {70}},\ \bibinfo {pages} {1895} (\bibinfo {year}
  {1993})}\BibitemShut {NoStop}%
\bibitem [{\citenamefont {Horodecki}\ \emph {et~al.}(2009)\citenamefont
  {Horodecki}, \citenamefont {Horodecki}, \citenamefont {Horodecki},\ and\
  \citenamefont {Horodecki}}]{horodecki2009quantum}%
  \BibitemOpen
  \bibfield  {author} {\bibinfo {author} {\bibfnamefont {R.}~\bibnamefont
  {Horodecki}}, \bibinfo {author} {\bibfnamefont {P.}~\bibnamefont
  {Horodecki}}, \bibinfo {author} {\bibfnamefont {M.}~\bibnamefont
  {Horodecki}}, \ and\ \bibinfo {author} {\bibfnamefont {K.}~\bibnamefont
  {Horodecki}},\ }\href@noop {} {\bibfield  {journal} {\bibinfo  {journal}
  {Reviews of modern physics}\ }\textbf {\bibinfo {volume} {81}},\ \bibinfo
  {pages} {865} (\bibinfo {year} {2009})}\BibitemShut {NoStop}%
\bibitem [{\citenamefont {Schr{\"o}dinger}(1935)}]{schrodinger1935discussion}%
  \BibitemOpen
  \bibfield  {author} {\bibinfo {author} {\bibfnamefont {E.}~\bibnamefont
  {Schr{\"o}dinger}},\ }in\ \href@noop {} {\emph {\bibinfo {booktitle}
  {Mathematical Proceedings of the Cambridge Philosophical Society}}},\
  Vol.~\bibinfo {volume} {31}\ (\bibinfo {organization} {Cambridge University
  Press},\ \bibinfo {year} {1935})\ pp.\ \bibinfo {pages}
  {555--563}\BibitemShut {NoStop}%
\bibitem [{\citenamefont {Dirac}(1981)}]{dirac1981principles}%
  \BibitemOpen
  \bibfield  {author} {\bibinfo {author} {\bibfnamefont {P.~A.~M.}\
  \bibnamefont {Dirac}},\ }\href@noop {} {\emph {\bibinfo {title} {The
  principles of quantum mechanics}}},\ \bibinfo {number} {27}\ (\bibinfo
  {publisher} {Oxford university press},\ \bibinfo {year} {1981})\BibitemShut
  {NoStop}%
\bibitem [{\citenamefont {Ichikawa}\ \emph {et~al.}(2008)\citenamefont
  {Ichikawa}, \citenamefont {Sasaki}, \citenamefont {Tsutsui},\ and\
  \citenamefont {Yonezawa}}]{ichikawa2008exchange}%
  \BibitemOpen
  \bibfield  {author} {\bibinfo {author} {\bibfnamefont {T.}~\bibnamefont
  {Ichikawa}}, \bibinfo {author} {\bibfnamefont {T.}~\bibnamefont {Sasaki}},
  \bibinfo {author} {\bibfnamefont {I.}~\bibnamefont {Tsutsui}}, \ and\
  \bibinfo {author} {\bibfnamefont {N.}~\bibnamefont {Yonezawa}},\ }\href@noop
  {} {\bibfield  {journal} {\bibinfo  {journal} {Physical Review A}\ }\textbf
  {\bibinfo {volume} {78}},\ \bibinfo {pages} {052105} (\bibinfo {year}
  {2008})}\BibitemShut {NoStop}%
\bibitem [{\citenamefont {Wei}(2010)}]{wei2010exchange}%
  \BibitemOpen
  \bibfield  {author} {\bibinfo {author} {\bibfnamefont {T.-C.}\ \bibnamefont
  {Wei}},\ }\href@noop {} {\bibfield  {journal} {\bibinfo  {journal} {Physical
  Review A}\ }\textbf {\bibinfo {volume} {81}},\ \bibinfo {pages} {054102}
  (\bibinfo {year} {2010})}\BibitemShut {NoStop}%
\bibitem [{\citenamefont {Ghirardi}(1977)}]{ghirardi1977gc}%
  \BibitemOpen
  \bibfield  {author} {\bibinfo {author} {\bibfnamefont {G.}~\bibnamefont
  {Ghirardi}},\ }\href@noop {} {\bibfield  {journal} {\bibinfo  {journal}
  {Nuovo Cimento B}\ }\textbf {\bibinfo {volume} {39}},\ \bibinfo {pages} {130}
  (\bibinfo {year} {1977})}\BibitemShut {NoStop}%
\bibitem [{\citenamefont {Ghirardi}\ and\ \citenamefont
  {Marinatto}(2004)}]{ghirardi2004general}%
  \BibitemOpen
  \bibfield  {author} {\bibinfo {author} {\bibfnamefont {G.}~\bibnamefont
  {Ghirardi}}\ and\ \bibinfo {author} {\bibfnamefont {L.}~\bibnamefont
  {Marinatto}},\ }\href@noop {} {\bibfield  {journal} {\bibinfo  {journal}
  {Physical Review A}\ }\textbf {\bibinfo {volume} {70}},\ \bibinfo {pages}
  {012109} (\bibinfo {year} {2004})}\BibitemShut {NoStop}%
\bibitem [{\citenamefont {Schliemann}(2001{\natexlab{a}})}]{schliemann2001j}%
  \BibitemOpen
  \bibfield  {author} {\bibinfo {author} {\bibfnamefont {J.}~\bibnamefont
  {Schliemann}},\ }\href@noop {} {\bibfield  {journal} {\bibinfo  {journal}
  {Phys. Rev. B}\ }\textbf {\bibinfo {volume} {63}},\ \bibinfo {pages} {085311}
  (\bibinfo {year} {2001}{\natexlab{a}})}\BibitemShut {NoStop}%
\bibitem [{\citenamefont {Schliemann}(2001{\natexlab{b}})}]{schliemann2001ja}%
  \BibitemOpen
  \bibfield  {author} {\bibinfo {author} {\bibfnamefont {J.}~\bibnamefont
  {Schliemann}},\ }\href@noop {} {\bibfield  {journal} {\bibinfo  {journal}
  {Phys. Rev. A}\ }\textbf {\bibinfo {volume} {64}},\ \bibinfo {pages} {022303}
  (\bibinfo {year} {2001}{\natexlab{b}})}\BibitemShut {NoStop}%
\bibitem [{\citenamefont {Eckert}(2002)}]{eckert2002k}%
  \BibitemOpen
  \bibfield  {author} {\bibinfo {author} {\bibfnamefont {K.}~\bibnamefont
  {Eckert}},\ }\href@noop {} {\bibfield  {journal} {\bibinfo  {journal} {Ann.
  Phys.(NY)}\ }\textbf {\bibinfo {volume} {299}},\ \bibinfo {pages} {88}
  (\bibinfo {year} {2002})}\BibitemShut {NoStop}%
\bibitem [{\citenamefont {Est{\`e}ve}(2008)}]{esteve2008j}%
  \BibitemOpen
  \bibfield  {author} {\bibinfo {author} {\bibfnamefont {J.}~\bibnamefont
  {Est{\`e}ve}},\ }\href@noop {} {\bibfield  {journal} {\bibinfo  {journal}
  {Nature (London)}\ }\textbf {\bibinfo {volume} {455}},\ \bibinfo {pages}
  {1216} (\bibinfo {year} {2008})}\BibitemShut {NoStop}%
\bibitem [{\citenamefont {Ghirardi}\ \emph {et~al.}(2002)\citenamefont
  {Ghirardi}, \citenamefont {Marinatto},\ and\ \citenamefont
  {Weber}}]{ghirardi2002entanglement}%
  \BibitemOpen
  \bibfield  {author} {\bibinfo {author} {\bibfnamefont {G.}~\bibnamefont
  {Ghirardi}}, \bibinfo {author} {\bibfnamefont {L.}~\bibnamefont {Marinatto}},
  \ and\ \bibinfo {author} {\bibfnamefont {T.}~\bibnamefont {Weber}},\
  }\href@noop {} {\bibfield  {journal} {\bibinfo  {journal} {Journal of
  Statistical Physics}\ }\textbf {\bibinfo {volume} {108}},\ \bibinfo {pages}
  {49} (\bibinfo {year} {2002})}\BibitemShut {NoStop}%
\bibitem [{\citenamefont {Paskauskas}\ and\ \citenamefont
  {You}(2001)}]{paskauskas2001r}%
  \BibitemOpen
  \bibfield  {author} {\bibinfo {author} {\bibfnamefont {R.}~\bibnamefont
  {Paskauskas}}\ and\ \bibinfo {author} {\bibfnamefont {L.}~\bibnamefont
  {You}},\ }\href@noop {} {\bibfield  {journal} {\bibinfo  {journal} {Phys.
  Rev. A}\ }\textbf {\bibinfo {volume} {64}},\ \bibinfo {pages} {042310}
  (\bibinfo {year} {2001})}\BibitemShut {NoStop}%
\bibitem [{\citenamefont {Zanardi}(2002)}]{zanardi2002p}%
  \BibitemOpen
  \bibfield  {author} {\bibinfo {author} {\bibfnamefont {P.}~\bibnamefont
  {Zanardi}},\ }\href@noop {} {\bibfield  {journal} {\bibinfo  {journal} {Phys.
  Rev. A}\ }\textbf {\bibinfo {volume} {65}},\ \bibinfo {pages} {042101}
  (\bibinfo {year} {2002})}\BibitemShut {NoStop}%
\bibitem [{\citenamefont {Shi}(2003)}]{shi2003shi}%
  \BibitemOpen
  \bibfield  {author} {\bibinfo {author} {\bibfnamefont {Y.}~\bibnamefont
  {Shi}},\ }\href@noop {} {\bibfield  {journal} {\bibinfo  {journal} {Phys.
  Rev. A}\ }\textbf {\bibinfo {volume} {67}},\ \bibinfo {pages} {024301}
  (\bibinfo {year} {2003})}\BibitemShut {NoStop}%
\bibitem [{\citenamefont {Barnum}\ \emph {et~al.}(2004)\citenamefont {Barnum},
  \citenamefont {Knill}, \citenamefont {Ortiz}, \citenamefont {Somma},\ and\
  \citenamefont {Viola}}]{barnum2004subsystem}%
  \BibitemOpen
  \bibfield  {author} {\bibinfo {author} {\bibfnamefont {H.}~\bibnamefont
  {Barnum}}, \bibinfo {author} {\bibfnamefont {E.}~\bibnamefont {Knill}},
  \bibinfo {author} {\bibfnamefont {G.}~\bibnamefont {Ortiz}}, \bibinfo
  {author} {\bibfnamefont {R.}~\bibnamefont {Somma}}, \ and\ \bibinfo {author}
  {\bibfnamefont {L.}~\bibnamefont {Viola}},\ }\href@noop {} {\bibfield
  {journal} {\bibinfo  {journal} {Physical Review Letters}\ }\textbf {\bibinfo
  {volume} {92}},\ \bibinfo {pages} {107902} (\bibinfo {year}
  {2004})}\BibitemShut {NoStop}%
\bibitem [{\citenamefont {Barnum}\ \emph {et~al.}(2005)\citenamefont {Barnum},
  \citenamefont {Ortiz}, \citenamefont {Somma},\ and\ \citenamefont
  {Viola}}]{barnum2005generalization}%
  \BibitemOpen
  \bibfield  {author} {\bibinfo {author} {\bibfnamefont {H.}~\bibnamefont
  {Barnum}}, \bibinfo {author} {\bibfnamefont {G.}~\bibnamefont {Ortiz}},
  \bibinfo {author} {\bibfnamefont {R.}~\bibnamefont {Somma}}, \ and\ \bibinfo
  {author} {\bibfnamefont {L.}~\bibnamefont {Viola}},\ }\href@noop {}
  {\bibfield  {journal} {\bibinfo  {journal} {International Journal of
  Theoretical Physics}\ }\textbf {\bibinfo {volume} {44}},\ \bibinfo {pages}
  {2127} (\bibinfo {year} {2005})}\BibitemShut {NoStop}%
\bibitem [{\citenamefont {Tichy}\ \emph {et~al.}(2013)\citenamefont {Tichy},
  \citenamefont {de~Melo}, \citenamefont {Ku{\'s}}, \citenamefont {Mintert},\
  and\ \citenamefont {Buchleitner}}]{tichy2013entanglement}%
  \BibitemOpen
  \bibfield  {author} {\bibinfo {author} {\bibfnamefont {M.~C.}\ \bibnamefont
  {Tichy}}, \bibinfo {author} {\bibfnamefont {F.}~\bibnamefont {de~Melo}},
  \bibinfo {author} {\bibfnamefont {M.}~\bibnamefont {Ku{\'s}}}, \bibinfo
  {author} {\bibfnamefont {F.}~\bibnamefont {Mintert}}, \ and\ \bibinfo
  {author} {\bibfnamefont {A.}~\bibnamefont {Buchleitner}},\ }\href@noop {}
  {\bibfield  {journal} {\bibinfo  {journal} {Fortschritte der Physik}\
  }\textbf {\bibinfo {volume} {61}},\ \bibinfo {pages} {225} (\bibinfo {year}
  {2013})}\BibitemShut {NoStop}%
\bibitem [{\citenamefont {Benatti}\ \emph {et~al.}(2011)\citenamefont
  {Benatti}, \citenamefont {Floreanini},\ and\ \citenamefont
  {Marzolino}}]{benatti2011entanglement}%
  \BibitemOpen
  \bibfield  {author} {\bibinfo {author} {\bibfnamefont {F.}~\bibnamefont
  {Benatti}}, \bibinfo {author} {\bibfnamefont {R.}~\bibnamefont {Floreanini}},
  \ and\ \bibinfo {author} {\bibfnamefont {U.}~\bibnamefont {Marzolino}},\
  }\href@noop {} {\bibfield  {journal} {\bibinfo  {journal} {Journal of Physics
  B: Atomic, Molecular and Optical Physics}\ }\textbf {\bibinfo {volume}
  {44}},\ \bibinfo {pages} {091001} (\bibinfo {year} {2011})}\BibitemShut
  {NoStop}%
\bibitem [{\citenamefont {Benatti}\ \emph
  {et~al.}(2012{\natexlab{a}})\citenamefont {Benatti}, \citenamefont
  {Floreanini},\ and\ \citenamefont {Marzolino}}]{benatti2012bipartite}%
  \BibitemOpen
  \bibfield  {author} {\bibinfo {author} {\bibfnamefont {F.}~\bibnamefont
  {Benatti}}, \bibinfo {author} {\bibfnamefont {R.}~\bibnamefont {Floreanini}},
  \ and\ \bibinfo {author} {\bibfnamefont {U.}~\bibnamefont {Marzolino}},\
  }\href@noop {} {\bibfield  {journal} {\bibinfo  {journal} {Annals of
  Physics}\ }\textbf {\bibinfo {volume} {327}},\ \bibinfo {pages} {1304}
  (\bibinfo {year} {2012}{\natexlab{a}})}\BibitemShut {NoStop}%
\bibitem [{\citenamefont {Benatti}\ \emph
  {et~al.}(2012{\natexlab{b}})\citenamefont {Benatti}, \citenamefont
  {Floreanini},\ and\ \citenamefont {Marzolino}}]{benatti2012entanglement}%
  \BibitemOpen
  \bibfield  {author} {\bibinfo {author} {\bibfnamefont {F.}~\bibnamefont
  {Benatti}}, \bibinfo {author} {\bibfnamefont {R.}~\bibnamefont {Floreanini}},
  \ and\ \bibinfo {author} {\bibfnamefont {U.}~\bibnamefont {Marzolino}},\
  }\href@noop {} {\bibfield  {journal} {\bibinfo  {journal} {Physical Review
  A}\ }\textbf {\bibinfo {volume} {85}},\ \bibinfo {pages} {042329} (\bibinfo
  {year} {2012}{\natexlab{b}})}\BibitemShut {NoStop}%
\bibitem [{\citenamefont {Marzolino}\ and\ \citenamefont
  {Buchleitner}(2015)}]{marzolino2015quantum}%
  \BibitemOpen
  \bibfield  {author} {\bibinfo {author} {\bibfnamefont {U.}~\bibnamefont
  {Marzolino}}\ and\ \bibinfo {author} {\bibfnamefont {A.}~\bibnamefont
  {Buchleitner}},\ }\href@noop {} {\bibfield  {journal} {\bibinfo  {journal}
  {Physical Review A}\ }\textbf {\bibinfo {volume} {91}},\ \bibinfo {pages}
  {032316} (\bibinfo {year} {2015})}\BibitemShut {NoStop}%
\bibitem [{\citenamefont {Marzolino}\ and\ \citenamefont
  {Buchleitner}(2016)}]{marzolino2016performances}%
  \BibitemOpen
  \bibfield  {author} {\bibinfo {author} {\bibfnamefont {U.}~\bibnamefont
  {Marzolino}}\ and\ \bibinfo {author} {\bibfnamefont {A.}~\bibnamefont
  {Buchleitner}},\ }\href@noop {} {\bibfield  {journal} {\bibinfo  {journal}
  {Proc. R. Soc. A}\ }\textbf {\bibinfo {volume} {472}},\ \bibinfo {pages}
  {20150621} (\bibinfo {year} {2016})}\BibitemShut {NoStop}%
\bibitem [{\citenamefont {Benatti}\ \emph {et~al.}(2017)\citenamefont
  {Benatti}, \citenamefont {Floreanini}, \citenamefont {Franchini},\ and\
  \citenamefont {Marzolino}}]{benatti2017remarks}%
  \BibitemOpen
  \bibfield  {author} {\bibinfo {author} {\bibfnamefont {F.}~\bibnamefont
  {Benatti}}, \bibinfo {author} {\bibfnamefont {R.}~\bibnamefont {Floreanini}},
  \bibinfo {author} {\bibfnamefont {F.}~\bibnamefont {Franchini}}, \ and\
  \bibinfo {author} {\bibfnamefont {U.}~\bibnamefont {Marzolino}},\ }\href@noop
  {} {\bibfield  {journal} {\bibinfo  {journal} {Open Systems \& Information
  Dynamics}\ }\textbf {\bibinfo {volume} {24}},\ \bibinfo {pages} {1740004}
  (\bibinfo {year} {2017})}\BibitemShut {NoStop}%
\bibitem [{\citenamefont {Franco}\ and\ \citenamefont
  {Compagno}(2016)}]{franco2016quantum}%
  \BibitemOpen
  \bibfield  {author} {\bibinfo {author} {\bibfnamefont {R.~L.}\ \bibnamefont
  {Franco}}\ and\ \bibinfo {author} {\bibfnamefont {G.}~\bibnamefont
  {Compagno}},\ }\href@noop {} {\bibfield  {journal} {\bibinfo  {journal}
  {Scientific reports}\ }\textbf {\bibinfo {volume} {6}},\ \bibinfo {pages}
  {20603} (\bibinfo {year} {2016})}\BibitemShut {NoStop}%
\bibitem [{\citenamefont {Sciara}\ \emph {et~al.}(2017)\citenamefont {Sciara},
  \citenamefont {Franco},\ and\ \citenamefont
  {Compagno}}]{sciara2017universality}%
  \BibitemOpen
  \bibfield  {author} {\bibinfo {author} {\bibfnamefont {S.}~\bibnamefont
  {Sciara}}, \bibinfo {author} {\bibfnamefont {R.~L.}\ \bibnamefont {Franco}},
  \ and\ \bibinfo {author} {\bibfnamefont {G.}~\bibnamefont {Compagno}},\
  }\href@noop {} {\bibfield  {journal} {\bibinfo  {journal} {Scientific
  Reports}\ }\textbf {\bibinfo {volume} {7}},\ \bibinfo {pages} {44675}
  (\bibinfo {year} {2017})}\BibitemShut {NoStop}%
\bibitem [{\citenamefont {Bellomo}\ \emph {et~al.}(2017)\citenamefont
  {Bellomo}, \citenamefont {Lo~Franco},\ and\ \citenamefont
  {Compagno}}]{bellomo2017}%
  \BibitemOpen
  \bibfield  {author} {\bibinfo {author} {\bibfnamefont {B.}~\bibnamefont
  {Bellomo}}, \bibinfo {author} {\bibfnamefont {R.}~\bibnamefont {Lo~Franco}},
  \ and\ \bibinfo {author} {\bibfnamefont {G.}~\bibnamefont {Compagno}},\
  }\href@noop {} {\bibfield  {journal} {\bibinfo  {journal} {Phys. Rev. A}\
  }\textbf {\bibinfo {volume} {96}},\ \bibinfo {pages} {022319} (\bibinfo
  {year} {2017})}\BibitemShut {NoStop}%
\bibitem [{\citenamefont {Compagno}\ \emph {et~al.}(2018)\citenamefont
  {Compagno}, \citenamefont {Castellini},\ and\ \citenamefont
  {Franco}}]{compagno2018dealing}%
  \BibitemOpen
  \bibfield  {author} {\bibinfo {author} {\bibfnamefont {G.}~\bibnamefont
  {Compagno}}, \bibinfo {author} {\bibfnamefont {A.}~\bibnamefont
  {Castellini}}, \ and\ \bibinfo {author} {\bibfnamefont {R.~L.}\ \bibnamefont
  {Franco}},\ }\href@noop {} {\bibfield  {journal} {\bibinfo  {journal} {Phil.
  Trans. R. Soc. A}\ }\textbf {\bibinfo {volume} {376}},\ \bibinfo {pages}
  {20170317} (\bibinfo {year} {2018})}\BibitemShut {NoStop}%
\bibitem [{\citenamefont {Lo~Franco}\ and\ \citenamefont
  {Compagno}(2018)}]{giuseppe2018}%
  \BibitemOpen
  \bibfield  {author} {\bibinfo {author} {\bibfnamefont {R.}~\bibnamefont
  {Lo~Franco}}\ and\ \bibinfo {author} {\bibfnamefont {G.}~\bibnamefont
  {Compagno}},\ }\href@noop {} {\bibfield  {journal} {\bibinfo  {journal}
  {Phys. Rev. Lett.}\ }\textbf {\bibinfo {volume} {120}},\ \bibinfo {pages}
  {240403} (\bibinfo {year} {2018})}\BibitemShut {NoStop}%
\bibitem [{\citenamefont {Castellini}\ \emph {et~al.}(2018)\citenamefont
  {Castellini}, \citenamefont {Bellomo}, \citenamefont {Compagno},\ and\
  \citenamefont {Franco}}]{castellini2018entanglement}%
  \BibitemOpen
  \bibfield  {author} {\bibinfo {author} {\bibfnamefont {A.}~\bibnamefont
  {Castellini}}, \bibinfo {author} {\bibfnamefont {B.}~\bibnamefont {Bellomo}},
  \bibinfo {author} {\bibfnamefont {G.}~\bibnamefont {Compagno}}, \ and\
  \bibinfo {author} {\bibfnamefont {R.~L.}\ \bibnamefont {Franco}},\
  }\href@noop {} {\bibfield  {journal} {\bibinfo  {journal} {arXiv preprint
  arXiv:1812.02141}\ } (\bibinfo {year} {2018})}\BibitemShut {NoStop}%
\bibitem [{\citenamefont {Louren{\c{c}}o}\ \emph {et~al.}(2019)\citenamefont
  {Louren{\c{c}}o}, \citenamefont {Debarba},\ and\ \citenamefont
  {Duzzioni}}]{lourencco2019entanglement}%
  \BibitemOpen
  \bibfield  {author} {\bibinfo {author} {\bibfnamefont {A.~C.}\ \bibnamefont
  {Louren{\c{c}}o}}, \bibinfo {author} {\bibfnamefont {T.}~\bibnamefont
  {Debarba}}, \ and\ \bibinfo {author} {\bibfnamefont {E.~I.}\ \bibnamefont
  {Duzzioni}},\ }\href@noop {} {\bibfield  {journal} {\bibinfo  {journal}
  {Physical Review A}\ }\textbf {\bibinfo {volume} {99}},\ \bibinfo {pages}
  {012341} (\bibinfo {year} {2019})}\BibitemShut {NoStop}%
\bibitem [{\citenamefont {Killoran}\ \emph {et~al.}(2014)\citenamefont
  {Killoran}, \citenamefont {Cramer},\ and\ \citenamefont
  {Plenio}}]{killoran2014extracting}%
  \BibitemOpen
  \bibfield  {author} {\bibinfo {author} {\bibfnamefont {N.}~\bibnamefont
  {Killoran}}, \bibinfo {author} {\bibfnamefont {M.}~\bibnamefont {Cramer}}, \
  and\ \bibinfo {author} {\bibfnamefont {M.~B.}\ \bibnamefont {Plenio}},\
  }\href@noop {} {\bibfield  {journal} {\bibinfo  {journal} {Physical review
  letters}\ }\textbf {\bibinfo {volume} {112}},\ \bibinfo {pages} {150501}
  (\bibinfo {year} {2014})}\BibitemShut {NoStop}%
\bibitem [{\citenamefont {Cavalcanti}\ \emph {et~al.}(2007)\citenamefont
  {Cavalcanti}, \citenamefont {Malard}, \citenamefont {Matinaga}, \citenamefont
  {Cunha},\ and\ \citenamefont {Santos}}]{cavalcanti2007useful}%
  \BibitemOpen
  \bibfield  {author} {\bibinfo {author} {\bibfnamefont {D.}~\bibnamefont
  {Cavalcanti}}, \bibinfo {author} {\bibfnamefont {L.}~\bibnamefont {Malard}},
  \bibinfo {author} {\bibfnamefont {F.}~\bibnamefont {Matinaga}}, \bibinfo
  {author} {\bibfnamefont {M.~T.}\ \bibnamefont {Cunha}}, \ and\ \bibinfo
  {author} {\bibfnamefont {M.~F.}\ \bibnamefont {Santos}},\ }\href@noop {}
  {\bibfield  {journal} {\bibinfo  {journal} {Physical Review B}\ }\textbf
  {\bibinfo {volume} {76}},\ \bibinfo {pages} {113304} (\bibinfo {year}
  {2007})}\BibitemShut {NoStop}%
\bibitem [{\citenamefont {Dalton}\ \emph
  {et~al.}(2017{\natexlab{a}})\citenamefont {Dalton}, \citenamefont {Goold},
  \citenamefont {Garraway},\ and\ \citenamefont {Reid}}]{dalton2017quantum}%
  \BibitemOpen
  \bibfield  {author} {\bibinfo {author} {\bibfnamefont {B.}~\bibnamefont
  {Dalton}}, \bibinfo {author} {\bibfnamefont {J.}~\bibnamefont {Goold}},
  \bibinfo {author} {\bibfnamefont {B.}~\bibnamefont {Garraway}}, \ and\
  \bibinfo {author} {\bibfnamefont {M.}~\bibnamefont {Reid}},\ }\href@noop {}
  {\bibfield  {journal} {\bibinfo  {journal} {Physica Scripta}\ }\textbf
  {\bibinfo {volume} {92}},\ \bibinfo {pages} {023004} (\bibinfo {year}
  {2017}{\natexlab{a}})}\BibitemShut {NoStop}%
\bibitem [{\citenamefont {Balachandran}\ \emph
  {et~al.}(2013{\natexlab{a}})\citenamefont {Balachandran}, \citenamefont
  {Govindarajan}, \citenamefont {de~Queiroz},\ and\ \citenamefont
  {Reyes-Lega}}]{balachandran2013entanglement}%
  \BibitemOpen
  \bibfield  {author} {\bibinfo {author} {\bibfnamefont {A.}~\bibnamefont
  {Balachandran}}, \bibinfo {author} {\bibfnamefont {T.}~\bibnamefont
  {Govindarajan}}, \bibinfo {author} {\bibfnamefont {A.~R.}\ \bibnamefont
  {de~Queiroz}}, \ and\ \bibinfo {author} {\bibfnamefont {A.}~\bibnamefont
  {Reyes-Lega}},\ }\href@noop {} {\bibfield  {journal} {\bibinfo  {journal}
  {Physical review letters}\ }\textbf {\bibinfo {volume} {110}},\ \bibinfo
  {pages} {080503} (\bibinfo {year} {2013}{\natexlab{a}})}\BibitemShut
  {NoStop}%
\bibitem [{\citenamefont {Balachandran}\ \emph
  {et~al.}(2013{\natexlab{b}})\citenamefont {Balachandran}, \citenamefont
  {Govindarajan}, \citenamefont {de~Queiroz},\ and\ \citenamefont
  {Reyes-Lega}}]{balachandran2013algebraic}%
  \BibitemOpen
  \bibfield  {author} {\bibinfo {author} {\bibfnamefont {A.}~\bibnamefont
  {Balachandran}}, \bibinfo {author} {\bibfnamefont {T.}~\bibnamefont
  {Govindarajan}}, \bibinfo {author} {\bibfnamefont {A.~R.}\ \bibnamefont
  {de~Queiroz}}, \ and\ \bibinfo {author} {\bibfnamefont {A.}~\bibnamefont
  {Reyes-Lega}},\ }\href@noop {} {\bibfield  {journal} {\bibinfo  {journal}
  {Physical Review A}\ }\textbf {\bibinfo {volume} {88}},\ \bibinfo {pages}
  {022301} (\bibinfo {year} {2013}{\natexlab{b}})}\BibitemShut {NoStop}%
\bibitem [{\citenamefont {Chin}\ and\ \citenamefont {Huh}()}]{Chin}%
  \BibitemOpen
  \bibfield  {author} {\bibinfo {author} {\bibfnamefont {S.}~\bibnamefont
  {Chin}}\ and\ \bibinfo {author} {\bibfnamefont {J.}~\bibnamefont {Huh}},\
  }\href@noop {} {\bibinfo  {journal} {in preparation}\ }\BibitemShut {NoStop}%
\bibitem [{Note1()}]{Note1}%
  \BibitemOpen
\bibfield  {journal} {  }\bibinfo {note} {This technic contrasts with NSA \cite
  {compagno2018dealing}, which performs all calculations with implicitly
  defined projection relations among states}\BibitemShut {NoStop}%
\bibitem [{\citenamefont {Tichy}\ and\ \citenamefont
  {M{\o}lmer}(2017)}]{tichy2017extending}%
  \BibitemOpen
  \bibfield  {author} {\bibinfo {author} {\bibfnamefont {M.~C.}\ \bibnamefont
  {Tichy}}\ and\ \bibinfo {author} {\bibfnamefont {K.}~\bibnamefont
  {M{\o}lmer}},\ }\href@noop {} {\bibfield  {journal} {\bibinfo  {journal}
  {Physical Review A}\ }\textbf {\bibinfo {volume} {96}},\ \bibinfo {pages}
  {022119} (\bibinfo {year} {2017})}\BibitemShut {NoStop}%
\bibitem [{\citenamefont {Paunkovic}(2004)}]{paunkovic2004role}%
  \BibitemOpen
  \bibfield  {author} {\bibinfo {author} {\bibfnamefont {N.}~\bibnamefont
  {Paunkovic}},\ }\emph {\bibinfo {title} {The role of indistinguishability of
  identical particles in quantum information processing}},\ \href@noop {}
  {Ph.D. thesis},\ \bibinfo  {school} {University of Oxford} (\bibinfo {year}
  {2004})\BibitemShut {NoStop}%
\bibitem [{\citenamefont {Peres}(2006)}]{peres2006quantum}%
  \BibitemOpen
  \bibfield  {author} {\bibinfo {author} {\bibfnamefont {A.}~\bibnamefont
  {Peres}},\ }\href@noop {} {\emph {\bibinfo {title} {Quantum theory: concepts
  and methods}}},\ Vol.~\bibinfo {volume} {57}\ (\bibinfo  {publisher}
  {Springer Science \& Business Media},\ \bibinfo {year} {2006})\BibitemShut
  {NoStop}%
\bibitem [{Note2()}]{Note2}%
  \BibitemOpen
  \bibinfo {note} {Both the concept of detection-level entanglement \cite
  {tichy2013entanglement} and the algebraic approach to the entanglement of
  indistinguishability \cite
  {benatti2011entanglement,benatti2012bipartite,benatti2017remarks} presume
  that the entanglement generated from indistinguishability depends on the
  experimental context, which implies not just the spatial overlap among
  identical particles but also between particles and dete ctors.}\BibitemShut
  {Stop}%
\bibitem [{\citenamefont {Baumgratz}\ \emph {et~al.}(2014)\citenamefont
  {Baumgratz}, \citenamefont {Cramer},\ and\ \citenamefont
  {Plenio}}]{baumgratz2014quantifying}%
  \BibitemOpen
  \bibfield  {author} {\bibinfo {author} {\bibfnamefont {T.}~\bibnamefont
  {Baumgratz}}, \bibinfo {author} {\bibfnamefont {M.}~\bibnamefont {Cramer}}, \
  and\ \bibinfo {author} {\bibfnamefont {M.}~\bibnamefont {Plenio}},\
  }\href@noop {} {\bibfield  {journal} {\bibinfo  {journal} {Physical review
  letters}\ }\textbf {\bibinfo {volume} {113}},\ \bibinfo {pages} {140401}
  (\bibinfo {year} {2014})}\BibitemShut {NoStop}%
\bibitem [{\citenamefont {Streltsov}\ \emph {et~al.}(2017)\citenamefont
  {Streltsov}, \citenamefont {Adesso},\ and\ \citenamefont
  {Plenio}}]{streltsov2017colloquium}%
  \BibitemOpen
  \bibfield  {author} {\bibinfo {author} {\bibfnamefont {A.}~\bibnamefont
  {Streltsov}}, \bibinfo {author} {\bibfnamefont {G.}~\bibnamefont {Adesso}}, \
  and\ \bibinfo {author} {\bibfnamefont {M.~B.}\ \bibnamefont {Plenio}},\
  }\href@noop {} {\bibfield  {journal} {\bibinfo  {journal} {Reviews of Modern
  Physics}\ }\textbf {\bibinfo {volume} {89}},\ \bibinfo {pages} {041003}
  (\bibinfo {year} {2017})}\BibitemShut {NoStop}%
\bibitem [{\citenamefont {Wiseman}\ and\ \citenamefont
  {Vaccaro}(2003)}]{wiseman2003entanglement}%
  \BibitemOpen
  \bibfield  {author} {\bibinfo {author} {\bibfnamefont {H.}~\bibnamefont
  {Wiseman}}\ and\ \bibinfo {author} {\bibfnamefont {J.~A.}\ \bibnamefont
  {Vaccaro}},\ }\href@noop {} {\bibfield  {journal} {\bibinfo  {journal}
  {Physical review letters}\ }\textbf {\bibinfo {volume} {91}},\ \bibinfo
  {pages} {097902} (\bibinfo {year} {2003})}\BibitemShut {NoStop}%
\bibitem [{\citenamefont {Amico}\ \emph {et~al.}(2008)\citenamefont {Amico},
  \citenamefont {Fazio}, \citenamefont {Osterloh},\ and\ \citenamefont
  {Vedral}}]{amico2008entanglement}%
  \BibitemOpen
  \bibfield  {author} {\bibinfo {author} {\bibfnamefont {L.}~\bibnamefont
  {Amico}}, \bibinfo {author} {\bibfnamefont {R.}~\bibnamefont {Fazio}},
  \bibinfo {author} {\bibfnamefont {A.}~\bibnamefont {Osterloh}}, \ and\
  \bibinfo {author} {\bibfnamefont {V.}~\bibnamefont {Vedral}},\ }\href@noop {}
  {\bibfield  {journal} {\bibinfo  {journal} {Reviews of modern physics}\
  }\textbf {\bibinfo {volume} {80}},\ \bibinfo {pages} {517} (\bibinfo {year}
  {2008})}\BibitemShut {NoStop}%
\bibitem [{\citenamefont {Chin}(2017)}]{chin2017coherence}%
  \BibitemOpen
  \bibfield  {author} {\bibinfo {author} {\bibfnamefont {S.}~\bibnamefont
  {Chin}},\ }\href@noop {} {\bibfield  {journal} {\bibinfo  {journal} {Physical
  Review A}\ }\textbf {\bibinfo {volume} {96}},\ \bibinfo {pages} {042336}
  (\bibinfo {year} {2017})}\BibitemShut {NoStop}%
\bibitem [{\citenamefont {Dalton}\ \emph
  {et~al.}(2017{\natexlab{b}})\citenamefont {Dalton}, \citenamefont {Goold},
  \citenamefont {Garraway},\ and\ \citenamefont {Reid}}]{dalton2017quantum2}%
  \BibitemOpen
  \bibfield  {author} {\bibinfo {author} {\bibfnamefont {B.}~\bibnamefont
  {Dalton}}, \bibinfo {author} {\bibfnamefont {J.}~\bibnamefont {Goold}},
  \bibinfo {author} {\bibfnamefont {B.}~\bibnamefont {Garraway}}, \ and\
  \bibinfo {author} {\bibfnamefont {M.}~\bibnamefont {Reid}},\ }\href@noop {}
  {\bibfield  {journal} {\bibinfo  {journal} {Physica Scripta}\ }\textbf
  {\bibinfo {volume} {92}},\ \bibinfo {pages} {023005} (\bibinfo {year}
  {2017}{\natexlab{b}})}\BibitemShut {NoStop}%
\bibitem [{\citenamefont {Cohen}(2016)}]{cohen2016aristotle}%
  \BibitemOpen
  \bibfield  {author} {\bibinfo {author} {\bibfnamefont {S.~M.}\ \bibnamefont
  {Cohen}},\ }\href@noop {} {\bibfield  {journal} {\bibinfo  {journal} {The
  Stanford Encyclopedia of Philosophy (Winter 2016 Edition), Edward N. Zalta
  (ed.)}\ } (\bibinfo {year} {2016})}\BibitemShut {NoStop}%
\end{thebibliography}%

\appendix

\section{Comparison of 1QL with NSA}\label{NSA}

It is straightforward to see that 1QL formalism presents quantitatively equivalent results to those of NSA, by considering Eq.~\eqref{ta} and \eqref{partsym}. First, the transition amplitude Eq.\eqref{ta} is rewritten as (up to normalization)
\begin{align}\label{A1}
&\<\Phi_1,\cdots,\Phi_N|\P_1,\cdots,\P_N\> \nn \\
& =\sum_{\r,\s} A_{\r(1)\s(1)}\cdots A_{\r(N)\s(N)} \nn \\
&= \sum_{\r,\s}A_{1,\r^{-1}\s(1)}\cdots A_{N,\r^{-1}\s(N)} \nn \\
&= =\sum_{\s'}A_{1\s'(1)}\cdots A_{N,\s'(N)}  
\end{align} where $A_{ij} \<\Phi_i|\P_j\>$. And the contraction of $\<\Phi|$ on $|\P_1,\cdots,\P_N\>$ gives
\begin{align} \label{A2}
&\<\Phi|\P_1,\cdots,\P_N\>\nn \\
&= \sum_{\s}\sum_{i} \<\Phi|_{A_i}\Big(|\P_1\>_{A_{\s(1)}}\cdots |\P_N\>_{A_{\s(N)}} \Big) \nn \\ 
&= \sum_{i}\<\Phi|\P_i\>\Big[\sum_{\s}  |\P_1\>_{A_{\s(1)}} \cdots (|\P_i\>_{A_{\s(i)}}) \cdots |\P_N\>_{A_{\s(N)}}\Big] \nn \\
&= \sum_{i}\<\Phi|\P_i\>|\P_1,\cdots,(\P_i), \cdots \P_N\>.  
\end{align}
where $(|\P_i\>_{A_{\s(i)}})$ in the secon equality and $(\P_i)$ in the last equality mean that $\P_i$ is absent. $|\P_1,\cdots,(\P_i), \cdots \P_N\>$ in the last equality follows the definition of Eq.~\eqref{partsym} with $n=N-1$. The equivalent relations with Eqs.~\eqref{A1} and \eqref{A2} can be found in Ref.~\cite{compagno2018dealing}, which presents the same protocol for obtaining a reduced density matrix.   

\section{Potentiality and actuality: the prerequisites for the entanglement of identical particles}\label{Potentiality}


Several works on the entanglement of identical particles have pointed out that this entanglement depends on both the initial particle state and the setup of detectors (which can be restated as the measurement process \cite{tichy2013entanglement} or the choice of observable (subalgebra) \cite{balachandran2013entanglement,balachandran2013algebraic,  benatti2017remarks}).  
To understand this situation intuitively, we should point out that for the particles to have nonzero coherence (relating to the detectors), they should overlap spatially. A set of spatially overlapping particles has a potential to be entangled, which however has a physical meaning (i.e., can be used as a resource) only after the entanglement is extracted onto distinguishable detectors \cite{wiseman2003entanglement,killoran2014extracting}. We can elucidate the process by borrowing the concept of \emph{potentiality-actuality dualism} from Aristotle (see, e.g., Ref.~\cite{cohen2016aristotle}). Before the particles arrive at the detectors, they only have the potentiality (the possibility that an object can have some feature) for entanglement. After the proper detection process, they have the actuality (the realization of the potentiality in the physical world) for entanglement.

\end{document}